\documentclass[a4paper,USenglish,cleveref, thm-restate, numberwithinsect]{lipics-v2021}
\usepackage[utf8]
{inputenc}

\title{Steiner Tree Parameterized by Multiway Cut and Even Less}

\nolinenumbers
\author{Bart M.P. Jansen}{Eindhoven University of Technology, The Netherlands}{b.m.p.jansen@tue.nl}{https://orcid.org/0000-0001-8204-1268
}{}
\author{C\'eline M.F. Swennenhuis}{Eindhoven University of Technology, The Netherlands}{cmfswennenhuis@gmail.com}{
https://orcid.org/0000-0001-9654-8094}{}

\funding{\flag[2cm]{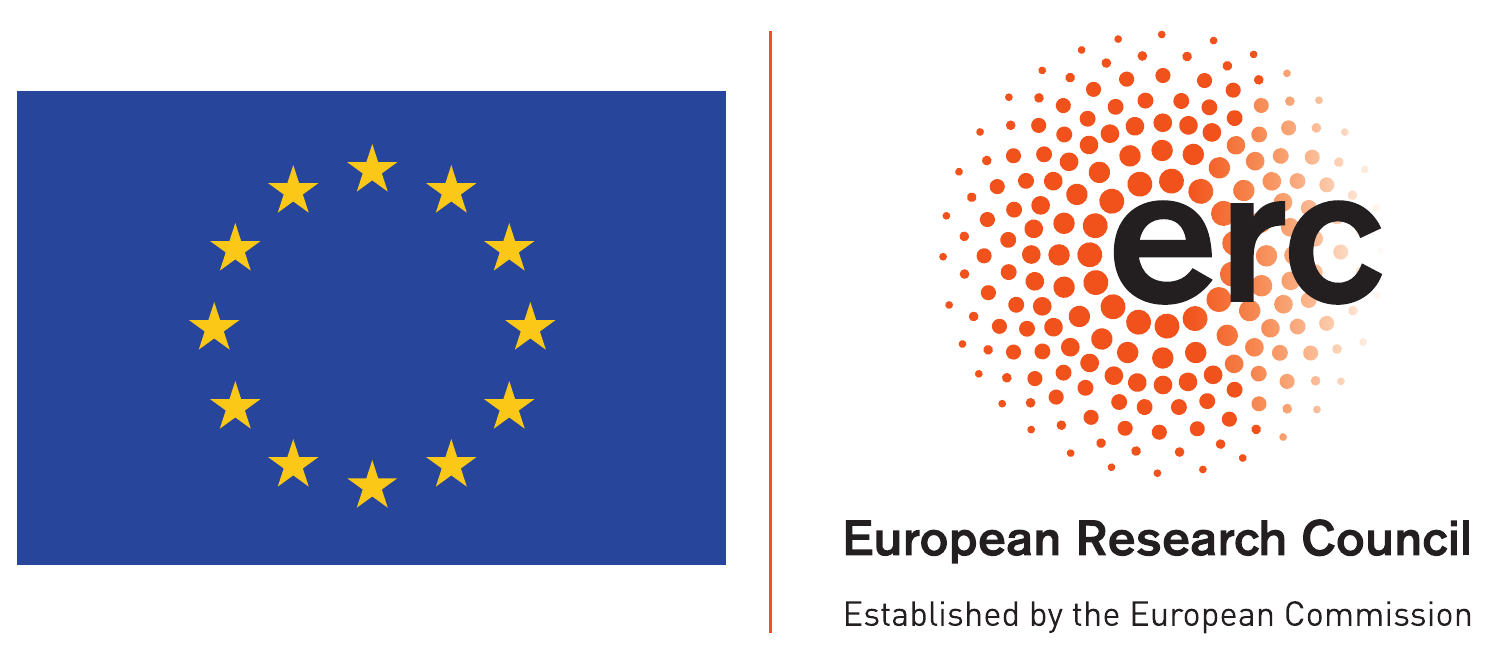}This project has received funding from the European Research Council (ERC) under the European Union's Horizon 2020 research and innovation programme (grant agreement No 803421, ReduceSearch).}

\Copyright{Bart M. P.~Jansen and C\'eline M.F.~Swennenhuis}

\authorrunning{B.M.P. Jansen and C.M.F. Swennenhuis} 

\EventEditors{Timothy Chan, Johannes Fischer, John Iacono, and Grzegorz Herman}
\EventNoEds{4}
\EventLongTitle{32nd Annual European Symposium on Algorithms (ESA 2024)}
\EventShortTitle{ESA 2024}
\EventAcronym{ESA}
\EventYear{2024}
\EventDate{September 2--4, 2024}
\EventLocation{Royal Holloway, London, United Kingdom}
\EventLogo{}
\SeriesVolume{308}
\ArticleNo{61}

\hideLIPIcs



\usepackage{amsmath, amssymb}
\usepackage{amsthm}
\usepackage[dvipsnames]{xcolor}
\usepackage{comment}
\usepackage{thmtools} 
\usepackage[linesnumbered,lined,boxed,commentsnumbered]{algorithm2e}

\def\DEBUG{true}
\ifdefined\DEBUG{}
\def\rem#1{{\marginpar{\raggedright\scriptsize #1}}}

\newcommand{\bmpr}[1]{\rem{\textcolor{orange}{\(\bullet \) #1}}}
\newcommand{\cs}[1]{{\color{blue}{#1}}}

\else

\newcommand{\bmpr}[1]{}
\newcommand{\cs}[1]{}

\fi

\newcommand{\Oh}{\mathcal{O}}

\newcommand{\hh}{\mathcal{H}}
\newcommand{\tw}{\mathsf{tw}}
\newcommand{\ed}{\mathsf{ed}}

\newcommand{\twtfree}{\tw_{{K}}}
\newcommand{\edkfree}{\ed_{{K}}}
\newcommand{\twdeltafree}{\tw_{\triangle\text{-}\mathrm{free}}}

\newcommand{\nat}{\mathbb{N}}
\newcommand{\T}{\mathbb{T}}

\newcommand{\KA}{\textnormal{\textbf{(K.A)}}\xspace}
\newcommand{\KB}{\textnormal{\textbf{(K.B)}}\xspace}
\newcommand{\KC}{\textnormal{\textbf{(K.C)}}\xspace}
\newcommand{\KD}{\textnormal{\textbf{(K.D)}}\xspace}

\newcommand{\DA}{\textnormal{\textbf{($\triangle$.A)}}\xspace}
\newcommand{\DB}{\textnormal{\textbf{($\triangle$.B)}}\xspace}
\newcommand{\DC}{\textnormal{\textbf{($\triangle$.C)}}\xspace}
\newcommand{\DD}{\textnormal{\textbf{($\triangle$.D)}}\xspace}

\newcommand{\poly}{\mathrm{poly}}
\newcommand{\cost}{\mathrm{cost}}

\usepackage{framed}
\usepackage[framemethod=TikZ]{mdframed}

\begin{document}

\maketitle

\begin{abstract}
In the \textsc{Steiner Tree} problem we are given an undirected edge-weighted graph as input, along with a set~$K$ of vertices called \emph{terminals}. The task is to output a minimum-weight connected subgraph that spans all the terminals. The famous Dreyfus-Wagner algorithm running in~$3^{|K|}\poly(n)$ time shows that the problem is fixed-parameter tractable parameterized by the number of terminals. We present fixed-parameter tractable algorithms for \textsc{Steiner Tree} using structurally smaller parameterizations.

Our first result concerns the parameterization by a multiway cut~$S$ of the terminals, which is a vertex set~$S$ (possibly containing terminals) such that each connected component of~$G-S$ contains at most one terminal. We show that \textsc{Steiner Tree} can be solved in~$2^{\Oh(|S|\log|S|)}\poly(n)$ time and polynomial space, where~$S$ is a minimum multiway cut for~$K$. The algorithm is based on the insight that, after guessing how an optimal Steiner tree interacts with a multiway cut~$S$, computing a minimum-cost solution of this type can be formulated as minimum-cost bipartite matching.

Our second result concerns a new hybrid parameterization called \emph{$K$-free treewidth} that simultaneously refines the number of terminals~$|K|$ and the treewidth of the input graph. By utilizing recent work on \textsc{$\mathcal{H}$-Treewidth} in order to find a corresponding decomposition of the graph, we give an algorithm that solves \textsc{Steiner Tree} in time~$2^{\Oh(k)} \poly(n)$, where~$k$ denotes the~$K$-free treewidth of the input graph. To obtain this running time, we show how the \emph{rank-based} approach for solving \textsc{Steiner Tree} parameterized by treewidth can be extended to work in the setting of $K$-free treewidth, by exploiting existing algorithms parameterized by~$|K|$ to compute the table entries of leaf bags of a tree $K$-free decomposition.
\end{abstract}

\clearpage

\section{Introduction}

\textsc{Steiner Tree} is a famous problem in algorithmic graph theory~\cite{DIMACS,survey,PACE}. In this problem, we are given an undirected edge-weighted graph~$G$ and a set~$K$ of \emph{terminal vertices} that we need to connect using edges of the graph. The goal is to find a minimum-weight connected subgraph that spans all these terminals, commonly known as a Steiner tree. The problem has a wide array of applications in industry, such as telecommunications, designing integrated circuits, molecular biology, and object detection (e.g.,~\cite{biologyideker2002discovering,circuitlengauer2012combinatorial,survey,telecommunication,objectrussakovsky2010steiner}).

Since the  \textsc{Steiner Tree} problem is~$\mathsf{NP}$-hard, we cannot hope to design an exact polynomial-time algorithm for this problem~\cite{Karp1972}. 
However, a popular approach is to bound the running time not just in terms of the input size~$n$, but to also take the influence of a secondary measurement~$k$ (referred to as the \emph{parameter}) into account.
In particular, we are interested in parameters that allow for fixed-parameter tractable (FPT) algorithms, i.e., algorithms that run in time~$f(k)\cdot \poly(n)$  where~$n$ is the size of the input,~$k$ the parameter, and~$f(\cdot)$ some computable function. One of the most natural parameters to consider for  \textsc{Steiner Tree} is~$|K|$, the number of terminals. The famous algorithm by Dreyfus and Wagner~\cite{dreyfus1971steiner} computes a minimum Steiner tree of an $n$-vertex graph in time~$3^{|K|}\cdot \poly(n)$ using exponential space, which was later improved by several authors~\cite{fuchs2007dynamic,lokshtanov2010saving}. In particular, Fomin, Kaski, Lokshtanov, Panolan, and Saurabh~\cite{fomin2019parameterized} present an algorithm that runs in single-exponential time~$7.97^{|K|}\cdot\poly(n)$ and uses \emph{polynomial} space. Another commonly-used parameter for \textsc{Steiner Tree} is the \emph{treewidth}~$\tw(G)$ of the input graph~$G$, which measures its structural similarity to a tree. A straight-forward dynamic program solves the problem in time~$2^{\Oh(\tw(G) \log \tw(G))} \poly(n)$. Using \emph{representative sets}, a single-exponential running time of $2^{\Oh(\tw(G))} \poly(n)$ can be obtained (cf.~\cite{cygan2022solving}).

For many other combinatorial problems on graphs, including graph modification problems such as \textsc{Vertex Cover} and \textsc{Odd Cycle Transversal}, a lot of effort has been invested into developing FPT algorithms using structurally \emph{smaller} parameterizations than standard measures of solution size or treewidth~\cite{Argawal2022,DistanceFromTriviality,HtreewidthEIBEN202157,distancetotriviality1guo2004structural,jansen2021vertex}. But to the best of our knowledge, no previous papers present FPT algorithms for \textsc{Steiner Tree} using parameterizations structurally smaller than the number of terminals. As our main contributions, we identify two relevant terminal-aware refined parameterizations for \textsc{Steiner Tree} and develop corresponding FPT algorithms.

\subparagraph*{Multiway cut} The first parameter we consider is the size of a (node) multiway cut for the terminals, i.e., a set of vertices~$S\subseteq V(G)$ such that every connected component of~$G-S$ contains at most one terminal. Since~$S$ is allowed to contain terminals, any instance has a multiway cut of size~$|K|-1$. In general, the size of a minimum multiway cut can be arbitrarily much smaller than~$|K|$. We show that we can solve \textsc{Steiner Tree} in FPT time and \emph{polynomial} space when we parameterize it by the size of a multiway cut for the terminals.

\begin{restatable}{theorem}{multiway}\label{thm:polyspace}
There is a polynomial-space algorithm that, given as input a graph~$G$ with weight function~$\cost \colon E(G) \to \mathbb{N}$, a set of terminals~$K \subseteq V(G)$, and a multiway cut~$S$ for~$K$, 
 outputs a minimum-weight Steiner tree in time~$2^{\Oh(|S|\log|S|)} \poly(n)$.
\end{restatable}

Our algorithm handles graphs whose weights are encoded in binary, as opposed to some other algorithms in the literature (\cite{lokshtanov2010saving,nederlof2013polyspace}) whose running time scales linearly with the value of the largest weight. The assumption that a multiway cut~$S$ for~$K$ is given as input is not restrictive, as a minimum multiway~$S$ cut for~$K$ can be found in~$4^{|S|}\poly(n)$ time and polynomial space\footnote{The algorithm from Chen, Liu, and Lu~\cite{MWCinFPT} is a bounded-depth branching algorithm, branching on important separators, making it a polynomial-space algorithm.}, due to an algorithm by Chen, Liu, and Lu~\cite{MWCinFPT}. \cref{thm:polyspace} improves upon an elementary \emph{exponential-space} algorithm with the same running time that was presented in a master's thesis supervised by the first author~\cite{Roozendaal23}.

\subparagraph*{$K$-free treewidth} The second parameterization we utilize refines the size of a multiway cut. To motivate the parameter, consider the following one-player game on a graph~$G$ with terminal vertices~$K$. In each round, one vertex is removed from each connected component. The game ends when each connected component contains at most one terminal. If there is a multiway cut~$S = \{s_1, \ldots, s_k\}$ of size~$k$, then the game can be won in at most~$k$ rounds by choosing vertex~$s_i$ in round~$i$ (or earlier, if not all remaining vertices of~$S$ belong to the same connected component). The minimum number of rounds needed is therefore never larger than the size of a minimum multiway cut. It can be arbitrarily much smaller: if deletions early in the game split the graph into multiple components, these are handled `in parallel' in subsequent rounds of the game.

For technical reasons, it will be convenient to consider the variation of the game that only ends when \emph{all} terminals have been removed from the graph, rather than ending when all terminals are separated. The number of rounds needed to win the latter variation is at most one more than the original: all terminals belong to different components when the original game ends, so at that point one additional round can delete one vertex from each connected component to eliminate all the terminals. Let~$\edkfree(G,K)$ (the \emph{elimination distance} to a $K$-free graph) denote the minimum number of rounds needed to eliminate all terminals from the graph. The previous discussion shows that if~$S$ is a multiway cut for~$K$ in~$G$, then we have~$\edkfree(G,K) \leq |S| + 1 \leq |K|$; hence it is a refined parameterization for \textsc{Steiner Tree}. The exponential dependence of our second algorithm can be bounded in terms of~$\edkfree(G,K)$. But there is an even smaller parameterization, called \emph{$K$-free treewidth}, that can be used to upper-bound the running time of the second algorithm we present. To introduce it, we briefly summarize an analogous range of parameterizations for vertex-deletion problems.

Our refined parameterizations for \textsc{Steiner Tree} are inspired by recent work on parameterized algorithms for \textsc{$\mathcal{H}$-Deletion}, which asks to find a minimum vertex set~$X$ in an input graph~$G$ that ensures~$G-X$ belongs to graph class~$\mathcal{H}$. The \textsc{Odd Cycle Transversal} problem is a prime example, which arises by letting~$\mathcal{H}$ be the class~$\mathsf{bip}$ of bipartite graphs. Recent work on \textsc{$\mathcal{H}$-Deletion}~\cite{Argawal2022,DistanceFromTriviality,HtreewidthEIBEN202157,jansen2021vertex,JansenKW23} has focused on improving parameterizations by the size of the deletion set~$X$ to parameterizations in terms of the \emph{elimination distance~$\ed_{\mathcal{H}}(G)$ to~$\mathcal{H}$}, which is the minimum number of rounds needed to obtain a graph in~$\mathcal{H}$ when removing one vertex from each connected component in each round. One of the results from this direction of work shows that \textsc{Odd Cycle Transversal} can be solved in time~$2^{\Oh(k)}\poly(n)$, where~$k = \ed_{\mathsf{bip}}(G)$~\cite{JansenKW23}. Through work of Bulian and Dawar~\cite{eliminationbulian2016graph}, it is known that the concept of elimination distance is related to the \emph{treedepth}~\cite{treedepth_nesetril2012bounded} of a graph~$G$: the treedepth is the minimum number of rounds needed to eliminate \emph{all} vertices. 

The famous graph parameter \emph{treewidth} is never larger than treedepth, but can be much smaller. Eiben et al.~\cite{HtreewidthEIBEN202157} proposed the notion of \emph{$\mathcal{H}$-treewidth}, where~$\mathcal{H}$ is a class of graphs. Roughly speaking, the $\mathcal{H}$-treewidth~$\tw_{\mathcal{H}}(G)$ of a graph~$G$ can be defined in terms of the minimum cost of a tree decomposition of a certain kind, in which (potentially large) leaf bags that induce a subgraph belonging to~$\mathcal{H}$ do not contribute to the cost. Hence~$\mathcal{H}$-treewidth captures how efficiently a graph can be decomposed into subgraphs belonging to~$\mathcal{H}$ along small separators in a treelike manner. It is known~\cite[Lemma 2.4]{jansen2021vertex} that~$\tw_{\mathcal{H}}(G) \leq \ed_{\mathcal{H}}(G)$ for all graphs~$G$, so that the resulting parameterization refines the $\mathcal{H}$-elimination distance. At the same time,~$\tw_{\mathcal{H}}(G)$ is not larger than the standard treewidth of~$G$, so that the resulting parameterization is a hybrid~\cite{DistanceFromTriviality} of standard treewidth and the solution size for $\mathcal{H}$-deletion. In their work on $\mathcal{H}$-treewidth, Jansen, de Kroon and W{\l}ordarczyk~\cite{jansen2021vertex} suggested that it may be interesting to explore variations of this parameter in the context of a set of terminal vertices. This is the route we pursue for \textsc{Steiner Tree}.

We complete the analogy between hybrid parameterizations for \textsc{$\mathcal{H}$-Deletion} and our parameterizations for \textsc{Steiner Tree} by introducing a notion called \emph{$K$-free treewidth}, denoted~$\twtfree(G,K)$ for a graph~$G$ with terminal set~$K$. While we defer formal definitions to \cref{sec:Kfreedecomp}, it intuitively captures how efficiently the input graph can be decomposed into terminal-free subgraphs along small separators in a tree-like manner. Using the correspondence between treewidth and treedepth, it follows that~$\twtfree(G,K) \leq \edkfree(G,K) \leq |S| + 1 \leq |K|$ where~$S$ is a minimum multiway cut. Our second main result shows that the resulting parameterization admits an algorithm whose running time and space usage are single-exponential.

\begin{restatable}{theorem}{terminaltwidth}
\label{thm:terminaltreewidth}
There is an algorithm that takes as input a graph~$G$ with weight function~$\cost \colon E(G) \to \mathbb{N}$, and a set~$K\subseteq V(G)$ of terminals, that computes a minimum-weight Steiner tree in~$2^{\Oh(\twtfree(G,K))}\poly(n)$ time and space.
\end{restatable}

\cref{thm:terminaltreewidth} shows that optimal solutions to \textsc{Steiner Tree} can still be found efficiently in inputs that can be decomposed into terminal-free (but potentially dense and large) subgraphs along small separators. The single-exponential dependence on~$\twtfree(G,K)$ in the running time of our algorithm is optimal under the \emph{Exponential Time Hypothesis}~\cite{ImpagliazzoP01}, and matches the single-exponential running times of the current-best algorithms parameterized by the number of terminals~$|K|$~\cite{fomin2019parameterized,nederlof2013polyspace} or treewidth~\cite{BODLAENDER201586}. Hence \cref{thm:terminaltreewidth} shows that the generality of the refined parameter~$\twtfree$ does not incur a significant computational overhead.

\subparagraph*{Techniques}

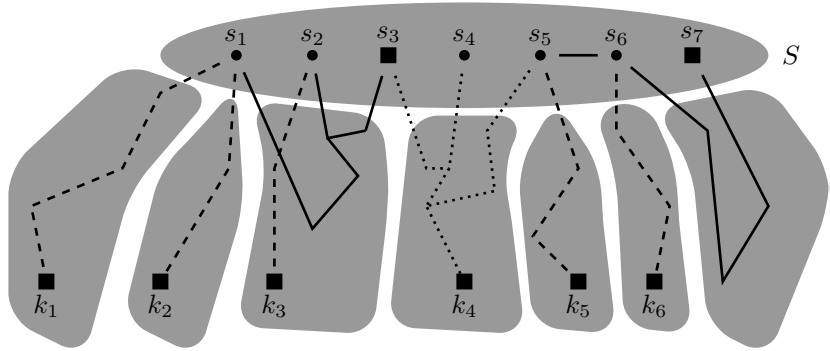
\begin{figure}[t]
    \centering
    \begin{tikzpicture} [scale = 1, 
    vertex/.style = {circle, draw, fill = black, align=center,minimum size = 3pt, inner sep = 0},
    terminal/.style = {rectangle, draw, fill = black, align=center,minimum size = 5pt, inner sep = 0}
,label distance=-2pt]

\tikzstyle{every path}=[line width=1pt]


\fill[fill = gray!80] (3,0) ellipse (4 and .7);

\fill[fill = gray!80, rounded corners = 10pt] (-3,-3.5) -- (-3,-1.5) -- (-1.5,-.1) -- (-.3,-.5) -- (-1.4,-1.8) -- (-2,-4) -- cycle;

\fill[fill = gray!80, rounded corners = 10pt] (-1.5,-3.5) -- (-1,-1.5) -- (-.5,-1) -- (0,-.4) -- (.1,-2) -- (-.5,-4) -- cycle;

\fill[fill = gray!80, rounded corners = 10pt] (0,-3.6) -- (.4,-1.5) -- (.2,-.6) -- (1.9,-.8) -- (2,-1.5) -- (2,-2) -- (1.8,-3.7) -- cycle;

\fill[fill = gray!80, rounded corners = 10pt] (2,-3.6) -- (2.3,-.8) -- (3.8,-.8) -- (3.5,-1.5) -- (3.6,-2) -- (3.8,-3.7) -- cycle;

\fill[fill = gray!80, rounded corners = 10pt] (3.9,-3.6) -- (3.7,-1.5) -- (4.2,-.6) -- (4.8,-1.5) -- (4.8,-2) -- (5,-3.7) -- cycle;

\fill[fill = gray!80, rounded corners = 10pt] (5.1,-3.6) -- (5,-1.6) -- (4.7,-.7) -- (5.4,-.6) -- (5.8,-1.5) -- (5.8,-2) -- (6,-3.7) -- cycle;

\fill[fill = gray!80, rounded corners = 10pt] (6.2,-3.6) -- (6,-1.5) -- (5.5,-.7) -- (7.,-.4) -- (7.8,-1.5) -- (7.8,-2) -- (7,-4) -- cycle;


\node[] at  (7.3,0) {$S$};

\node[vertex,label=above:$s_1$] (s_1) at  (0,0) {};
\node[vertex,label=above:$s_2$] (s_2) at  (1,0) {};
\node[terminal,label=above:$s_3$] (s_3) at  (2,0) {};
\node[vertex,label=above:$s_4$] (s_4) at  (3,0) {};
\node[vertex,label=above:$s_5$] (s_5) at  (4,0) {};
\node[vertex,label=above:$s_6$] (s_6) at  (5,0) {};
\node[terminal,label=above:$s_7$] (s_7) at  (6,0) {};

\node[terminal,label=below:$k_1$] (k_1) at  (-2.5,-3) {};
\node[terminal,label=below:$k_2$] (k_2) at  (-1,-3) {};
\node[terminal,label=below:$k_3$] (k_3) at  (.5,-3) {};
\node[terminal,label=below:$k_4$] (k_4) at  (3,-3) {};
\node[terminal,label=below:$k_5$] (k_5) at  (4.5,-3) {};
\node[terminal,label=below:$k_6$] (k_6) at  (5.5,-3) {};

\draw[dashed, shorten <= 5pt] (s_1) -- (-1,-.5) -- (-1.5,-1.5) -- (-2.7, -2) -- (k_1);
\draw[dashed, shorten <= 5pt] (s_1) -- (-.1,-1.5) -- (k_2);
\draw[dashed, shorten <= 5pt] (s_2) -- (.5,-1.5) -- (k_3);
\draw[dashed, shorten <= 5pt] (s_5) -- (4.5,-1.5) -- (3.9, -2.4) --  (k_5);
\draw[dashed, shorten <= 5pt] (s_6) -- (5,-1) -- (5.7,-2) -- (k_6);

\draw[dotted, shorten <= 5pt] (s_3) -- (2.5,-1.5) -- (2.8,-1.5);
\draw[dotted, shorten <= 5pt] (s_4) -- (2.8,-1.5) -- (2.5,-2);
\draw[dotted, shorten <= 5pt] (s_5) -- (3.3,-1) -- (3.4, -1.8) -- (2.5,-2);
\draw[dotted] (k_4) -- (2.5,-2);

\draw[shorten <= 5pt, shorten >= 5pt] (s_6) -- (6.2,-1) -- (6.4, -3) -- (7,-2) -- (s_7);
\draw[shorten <= 5pt, shorten >= 5pt] (s_6) -- (s_5);
\draw[shorten <= 5pt] (s_1) -- (1, -2.3) -- (1.6,-1.6) -- (1.2,-1.1) ;
\draw[shorten <= 5pt] (s_2) -- (1.2,-1.1) -- (1.7,-1);
\draw[shorten <= 5pt] (s_3) -- (1.7,-1);


\end{tikzpicture}
    \caption{The multiway cut $S= \{s_1,s_2,s_3,s_4,s_5,s_6,s_7\}$ ensures that no two terminals (squares) belong to the same connected component of $G-S$. The gray areas indicate the set $S$ and the components of $G-S$. The lines are a visual representation of a Steiner tree $F$ for the terminals, split at $S$, such that $F$ is partitioned into three trees of category~1 (solid trees, note the solid tree between $s_5$ and $s_6$), one tree of category~2 (dotted trees) and five paths of category~3 (dashed paths).}
    \label{fig:splitting}
\end{figure}

Both of our algorithms utilize a subroutine to solve \textsc{Steiner Tree} for a small set of terminal vertices. The manner in which these subroutine results are employed differs greatly between the two, however.

We first sketch the main idea behind~\cref{thm:polyspace}, which aims to find a minimum Steiner tree by utilizing a known multiway cut~$S$. We can partition the edges of any minimum Steiner tree~$F$ by splitting~$F$ at the vertices of the multiway cut~$S$ (see~\cref{fig:splitting}), such that any tree~$F'$ in the resulting partition intersects with at most one component of~$G-S$ and falls into one of the following categories: \begin{enumerate}
    \item~$F'$ is a minimum-weight Steiner tree for a subset of~$S$,
    \item~$F'$ is a minimum-weight Steiner tree for a subset of~$S$ and exactly one~$k\in K$, or 
    \item~$F'$ is a minimum-weight path from a terminal~$k \in K$ to some vertex of~$S$.
\end{enumerate}  

There can never be two or more vertices from~$K\setminus S$ in a single tree~$F'$, as these terminals live in different components of~$G-S$. 
Note that the paths of category 3 can be found quickly by computing a shortest path. Moreover, we prove that the Steiner trees of categories 1 and 2 are Steiner trees for a vertex set of size at most~$|S|+1$, hence we can compute the minimum weights of such Steiner trees in polynomial space with the algorithm by Fomin, Kaski, Lokshtanov, Panolan, and Saurabh~\cite{fomin2019parameterized}. 

We introduce the notion of~\emph{$S$-connecting system} to characterize how a Steiner tree~$F$ interacts with the set~$S$. Intuitively, an $S$-connecting system records which types of trees~$F'$ (in terms of the classification above) arise when splitting~$F$ at~$S$, and which vertices from~$S \cup K$ are connected by these trees. The algorithm will iterate over all~$2^{\Oh(|S|\log|S|)}$ distinct~$S$-connecting systems. For each such system, we can compute an edge-weighted bipartite graph~$B$ such that an optimal Steiner tree consistent with the $S$-connecting system corresponds to a minimum-weight maximum matching in~$B$. The edges of~$B$ correspond to the weights of optimal Steiner trees corresponding to either category 1 or 2. By exploiting the fact that each tree~$F'$ obtained by splitting an optimal tree~$F$ at~$S$ involves only few terminals, each bipartite graph~$B$ can be computed in single-exponential time.

For Theorem~\ref{thm:terminaltreewidth}, the algorithm builds upon the~$2^{\Oh(\tw(G))}\poly(n)$ time algorithm from Bodlaender, Cygan, Kratsch, and Nederlof~\cite{BODLAENDER201586}. Recall that any node~$x$ of a rooted tree decomposition can be associated with a subgraph~$G_x$, which contains all subgraphs of its children. Note that its bag~$\chi(x)\subseteq V(G)$ is the \emph{boundary} of~$G_x$, i.e., the only vertices of~$G_x$ that can have neighbors outside~$G_x$ belong to~$\chi(x)$. The algorithm goes over the tree decomposition, keeping track of possible \emph{partial solutions}, which are Steiner trees restricted to the subgraph~$G_x$. Instead of storing actual partial solutions, the algorithm stores how these partial solution are connected to the boundary in the form of partitions. In principle, this could lead to~$\tw(G)^{\Oh(\tw(G))}$ different partitions that would need to be stored as any bag~$\chi(x)$ is of size at most~$\tw(G) +1$ by definition. However, using the rank-based approach~\cite{BODLAENDER201586}, one only needs to keep a \emph{representative set} of partitions of size at most~$2^{\tw(G)}$; we sketch the main ideas in the following paragraph. 

To obtain an algorithm parameterized by~$\twtfree(G,K)$, for the parts of the decomposition corresponding to small separators we can use the same techniques as employed for standard treewidth. For the parts of the decomposition that are large, the decomposition ensures us that terminals can only lie on the boundary. Hence, there are only~$\twtfree(G,K) +1$ terminals in~$G_x$ for such nodes~$x$ and we can use the Dreyfus-Wagner algorithm~\cite{dreyfus1971steiner} to obtain a fixed-parameter tractable algorithm to compute partial solutions of~$G_x$. Since there might be~$\twtfree(G,K)^{\Oh(\twtfree(G,K))}$ partial solutions, this does not directly yield a single-exponential running time. To obtain Theorem~\ref{thm:terminaltreewidth}, we apply the rank-based approach for these nodes, by iteratively increasing the set of vertices of the boundary that are used by considered partial solutions. As far as we know, our algorithm is the first to incorporate advanced dynamic-programming ideas such as the rank-based approach with hybrid graph decompositions. To obtain a decomposition to which we can apply this scheme, we show that a recent FPT 5-approximation to compute $\mathcal{H}$-treewidth~\cite{JansenKW23} can be leveraged for $K$-free treewidth.

\subparagraph*{Organization}
We give the polynomial-space algorithm of Theorem~\ref{thm:polyspace} in Section~\ref{sec:polyspace}. In Section~\ref{sec:Kfreedecomp} we formally define~$K$-free treewidth and give an FPT 5-approximation. In Section~\ref{sec:singleexpalgorithm} we prove Theorem~\ref{thm:terminaltreewidth}, i.e., we give a single-exponential algorithm for \textsc{Steiner Tree} when parameterized by~$K$-free treewidth. We conclude in \cref{sec:conclusion}. 

We use standard terminology for graphs and parameterized algorithms. Terms not defined here can be found in a textbook~\cite{cygan2015parameterized}. We write~$[m]$ for~$\{1,...,m\}$ and define~$\min\{\emptyset\} = \infty$. 

\section{Polynomial-space algorithm parameterized by multiway cut}
\label{sec:polyspace}
In this section, we will consider \textsc{Steiner Tree} parameterized by the size of a given multiway cut~$S$ for the terminal set~$K$. In other words, each connected component of~$G-S$ contains at most one terminal. Note that $S$ can contain terminals. 

\subsection{$S$-connecting systems} \label{subsec:Sconnecting}

The following concept is the main focus of this section.

\begin{definition}[$S$-connecting system] \label{def:s:connecting} Consider a graph $G$ and $S \subseteq V(G)$. An \emph{$S$-connecting system} in~$G$ is a tuple $(\mathcal{S}, \mathcal{T})$, where $\mathcal{S} = \{S_1, S_2, \dots, S_{m}\}$ is a collection of subsets of $S$ and $\mathcal{T}$ is a tree, such that: 
\begin{enumerate}
    \item $V(\mathcal{T}) = S \cup \{u_1, \dots, u_{m}\}$, for $m = |\mathcal{S}|$; \label{cond:system:vtxset}
    \item for all $i \in [m]$ we have $S_i = N_{\mathcal{T}}(u_i) \subseteq S$ and $d_{\mathcal{T}}(u_i) > 1$; and\label{cond:system:t}
    \item for all distinct~$s, s' \in S$ it holds that if~$\{s,s'\} \in E(\mathcal{T})$, then~$\{s, s'\} \in E(G)$. \label{cond:edge:subgraph}
    \end{enumerate}
\end{definition}

We will use the following notion of self-reachable to describe that vertices are part of a common connected component.
\begin{definition}[Self-reachable] Consider a graph $G$ and $S \subseteq V(G)$. The set $S$ is \emph{self-reachable} in $G$ if $S$ is contained in a single connected component of $G$, i.e., when there is a path in~$G$ between any pair of vertices $s,s' \in S$.
\end{definition}

Intuitively, if~$S$ is a multiway cut for terminal set~$K$ in graph~$G$, and~$F$ is a Steiner tree for~$K$, then there is an $S$-connecting system~$\mathcal{S}$ that represents the interaction between the Steiner tree~$F$ and the multiway cut~$S$. The tree~$\mathcal{T}$ of this $S$-connecting system can be obtained from~$F$ by contracting each tree of~$F-S$ into a single vertex, while removing the degree-1 vertices of the resulting tree that do not belong to~$S$. \cref{def:s:connecting} ensures the corresponding set~$\mathcal{S}$ is uniquely determined by~$\mathcal{T}$. We say that the tree~$F$ \emph{realizes} the resulting $S$-connecting system~$(\mathcal{S}, \mathcal{T})$.

The following lemma follows from the definitions by a simple induction. It shows that when~$H$ is a subgraph of~$G$ in which each set~$S_i$ of an $S$-connecting system~$(\mathcal{S} = \{S_1, \ldots, S_m\}, \mathcal{T})$ is self-reachable, then the connectivity structure of the tree~$\mathcal{T}$ ensures that the entire set~$S$ is self-reachable.

\begin{restatable}{lemma}{lemConnectedWithT}  \label{lem:connectedwithT}
    Let $(\mathcal{S}, \mathcal{T})$ be an $S$-connecting system and let $H \subseteq G$ be a subgraph of $G$ such that for all $i \in [m]$, the set $S_i$ is self-reachable in $H$. Then $S$ is self-reachable in $H' = H \cup \mathcal{T}[S]$.
\end{restatable}
\begin{proof}
    We prove that any pair~$\{s,s'\}$ of vertices from~$S$ is self-reachable in~$H$, by induction on the distance from~$s$ to~$s'$ in~$\mathcal{T}$. Note that the base case is evidently true, as for distance $0$ we have $s = s'$.  Consider an arbitrary pair $\{s,s'\}$ and let~$\ell$ be the distance between them in~$\mathcal{T}$; the induction hypothesis is that any pair from~$S$ whose distance in~$\mathcal{T}$ is less than~$\ell$, is self-reachable in~$H'$.    
    
    Let $s = x_0, x_1, \dots, x_{\ell} = s'$ be a path in $\mathcal{T}$ from $s$ to $s'$.
    If $x_{\ell-1} \in S$, we find that by induction that $\{s,x_{\ell-1}\}$ is self-reachable in $H'$.
    Moreover, $\{x_{\ell-1},s'\}\in E(\mathcal{T}[S])$ so $\{s'',s'\}\in E(H')$ and we find that $\{x_{\ell-1},s'\}$ is self-reachable in $H'$.
    Hence, $\{s,s'\}$ is self-reachable in $H'$. 
    
    If $x_{\ell-1} \not\in S$, then by definition of $S$-connecting system, we see that $x_{\ell-2}\in S$ as $x_{\ell-1}$ can then only contain vertices from $S$ as its neighbors. Moreover, since we have $x_{\ell-2}, s' \in N(x_{\ell-1})$ we find that $x_{\ell-2}, s' \in S_i$ for some $i \in [m]$. By assumption, $S_i$ is self-reachable in $H'$, so in particular $\{x_{\ell-2},s'\}$ is self-reachable in $H'$. 
    Again by induction we find that $\{s,x_{\ell-2}\}$ is self-reachable in $H'$. Hence, $\{s,s'\}$ is self-reachable in $H'$. 
\end{proof}

Next, we will prove that there are at most $2^{\Oh(|S|\log|S|)}$ different $S$-connecting systems. For this, we first prove the following claim, bounding the size of $\mathcal{S}$.

\begin{lemma}
    \label{lemma:boundedm} 
    Let~$G$ be a graph and~$S\subseteq V(G)$. Any $S$-connecting system $(\mathcal{S},\mathcal{T})$ in~$G$ satisfies~$|\mathcal{S}| \leq |S| - 1$. 
\end{lemma} 
\begin{proof}
Let~$(\mathcal{S}, \mathcal{T})$ be an $S$-connecting system. Let~$\mathcal{S} = \{S_1, \dots, S_m\}$; we prove that~$m \leq |S| - 1$. 
Let $V(\mathcal{T}) = S \cup \{u_1, \dots, u_{m}\}$ such that $N_{\mathcal{T}}(u_i) = S_i$ for all $i \in [m]$. 
Root $\mathcal{T}$ at an arbitrary vertex ~$s^* \in S$. Note that for any $i \in [m]$, $u_i$ is not a leaf of $\mathcal{T}$ and has only neighbors in $S$ by definition of the $S$-connecting system. Hence, every $u_i$ is a parent of at least one $s \in S \setminus \{s^*\}$. Since every $s \in S$ has at most one parent we find $m \le |S| -1$. 
\end{proof}

\begin{lemma}\label{lemma:nrScon}
For any graph $G$ and vertex set $S \subseteq V(G)$, there are at most $2^{\Oh(|S|\log|S|)}$ different $S$-connecting systems.
\end{lemma}
\begin{proof}
Consider an arbitrary~$S$-connecting system~$(\mathcal{S}, \mathcal{T})$ in~$G$. \cref{cond:system:vtxset} of \cref{def:s:connecting} ensures that~$|V(\mathcal{T})| \leq |S| + |\mathcal{S}|$, while \cref{lemma:boundedm} ensures~$|\mathcal{S}| \leq |S|-1$. Hence $|V(\mathcal{T})| \le 2|S| -1$. By Cayley's formula, there are $n^{n-2}$ different labeled trees on $n$ vertices. Therefore, the number of different choices for~$\mathcal{T}$ is bounded by 
$\sum_{i=1}^{2|S|-1}i^{i-2} \le (2|S|-1) \left(2|S| -1\right)^{2|S| -3}$, i.e., by $2^{\Oh(|S| \log |S|)}$.
By \cref{def:s:connecting}, the collection~$\mathcal{S}$ is uniquely determined by~$\mathcal{T}$. 
\end{proof}

\subsection{The algorithm} \label{subsec:Algorithm}

Before we present \cref{thm:polyspace}, we describe at a high level how $S$-connecting systems facilitate a reduction to bipartite matching. The starting observation is that, by trying all possible subsets of the multiway cut~$S$, we may assume that the Steiner tree we are looking for contains all vertices of~$S$. The connectivity pattern of each Steiner tree with respect to~$S$ can then be summarized by an $S$-connecting system. To assemble a Steiner tree that realizes a given $S$-connecting system, a na\"ive approach is the following: pick a shortest path from each terminal to~$S$, and for each subset~$S_i$ of the $S$-connecting system, pick a subtree containing $S_i$ to make~$S_i$ self-reachable. This approach leads to some redundancy: we might be able to re-use some edges if the path used to connect a terminal~$t_p$ to~$S$, shares some edges with a subtree that makes a subset~$S_i$ self-reachable. 

Due to~$S$ being a multiway cut, a tree of~$G-S$ that makes a subset~$S_i$ self-reachable can live in only one component of~$G-S$, and can therefore only involve at most one terminal. For each choice of terminal~$t_p$ and subset~$S_i$, we can use the polynomial-space FPT algorithm for \textsc{Steiner Tree} parameterized by~$|K$|~\cite{fomin2019parameterized} to compute the cost of a tree making~$\{t_p\} \cup S_i$ self-reachable, and compare it to the shortest-path distance between $t_p$ and the closest vertex of $S$ to see how much we would benefit from combining the task of making~$t_p$ reachable from~$S$ with the task of making~$S_i$ self-reachable. As each terminal can only be involved in making one set $S_i$ self-reachable, we now see a weighted bipartite matching problem appear: we are looking for a pairing of sets~$S_i$ with terminals~$t_p$ so that the overall \emph{savings}, compared to making each terminal individually reachable from $S$ by a shortest path and separately adding a subtree making each~$S_i$ self-reachable into the Steiner tree, are as large as possible.

These ideas are formalized in the proof of the following theorem. 
\multiway*

\begin{proof}

Consider the multiway cut~$S$. Let~$C_1, \ldots, C_q$ denote the vertex sets of the connected components of~$G-S$, so that each~$C_p$ contains at most one terminal. For a vertex set~$X$, we denote by~$\delta(X)$ the set of edges of~$G$ that have exactly one endpoint inside~$X$. Hence for~$p \in [q]$, each edge of~$\delta(C_p)$ has one endpoint in~$C_p$ and one endpoint in~$S$. We define $k_p := C_p \cap K$ for all~$p \in [q]$. 

In our algorithm, we will guess  which vertices $S' \subseteq S$ will be used by the Steiner tree. 
For each such guess for $S'$, we will consider all possible $S'$-connecting systems $(\mathcal{S},\mathcal{T})$. For each such system, we create a weighted complete bipartite graph $B$ with partition $V(B) = \{\mathcal{S}\} \cup \{ P \}$ where
$P = \{ (p,j) \colon p \in \{1,\dots,q\}, j \in \{0,\dots,|\mathcal{S}|\}\}$.
For each $p \in [q]$ we define $G_p$ as $G[C_p] \cup \delta(C_p)$, i.e., the graph induced on $C_p$ together with the edges between $C_p$ and $S$. 
The goal is to find a minimum weight maximum matching in $B$, which represents how each set $S_i \in \mathcal{S}$ is self-reachable. If $S_i$ is matched with $(p,0)$, it indicates that the subtree making $S_i$ self-reachable is contained in $G_p$, and that $k_p$ is contained in this subtree. If $S_i$ is matched with $(p,j)$ for $j>0$, it indicates that the subtree making $S_i$ self-reachable is contained in $G_p$ (not necessarily using terminal $k_p$ if it exists). Note that the specific value of $j$ does not carry any meaning, but we need to be able to use $G_p$ multiple (and at most $|\mathcal{S}|$) times to make different sets in $\mathcal{S}$ self-reachable. 

To determine the weights of $B$, we have to solve several  \textsc{Steiner Tree} problems in succession. Let $\mathsf{MST}[H,X]$ denote an arbitrary minimum-cost Steiner tree for terminal set $X$ in graph $H$ and let $\mathsf{MST}[H,X] = \emptyset$ if no such Steiner tree exists. 
Moreover, for $X \subseteq V(G)$ and 
$k \in V(G)$ we refer to $\mathsf{SP}[X,k]$ as an arbitrary shortest path from $k$ to some vertex of $X$ and we set $\mathsf{SP}[X,k] = \emptyset$ is no such path exists. We extend this notation for the sets~$k_p$ defined above, which are either singletons or empty. If~$k_p = \{t_p\}$ is a singleton set containing a terminal, we let~$\mathsf{SP}[X,k_p] = \mathsf{SP}[X, t_p]$. If~$k_p = \emptyset$, then~$\mathsf{SP}[X,k_p] = \emptyset$. 

Some of the edges of $B$ can have infinite weight, indicating that a certain Steiner tree does not exist.  We define a cost function~$\omega$ for the edges of~$B$ as follows. For all $S_i \in \mathcal{S}$ and $p \in [q]$ we set 
\begin{align*}
    \omega(S_i, (p,0)) &= \begin{cases}
    \infty &\text{ if $\mathsf{SP}[S',k_p] = \emptyset$},\\
\cost (\mathsf{MST}[G_p,S_i\cup \{k_p\}]) - \cost(\mathsf{SP}[S',k_p]) &\text{ else if $C_p \cap K \not = \emptyset$},\\
\infty & \text{ otherwise,}
    \end{cases}
\end{align*}
and for all $S_i \in \mathcal{S}$, $p \in [q]$, and $j \in [|\mathcal{S}|]$ we set
\begin{align*}
    \omega(S_i, (p,j)) &= \cost(\mathsf{MST}[G_p,S_i]).
\end{align*}

Note that each of these values can be computed in polynomial space and $2^{\Oh(|S'|)}\poly(n)$ time using the algorithm by Fomin, Kaski, Lokshtanov, Panolan, and Saurabh~\cite{fomin2019parameterized}. Since the bipartite graph $B$ is of polynomial size, we can construct $B$ in polynomial space and $2^{\Oh(|S'|)}\poly(n)$ time.
We claim that we can reconstruct a Steiner tree, based on a minimum weight maximum matching of $B$, if its weight is finite.

\begin{claim}
    \label{claim:MtoT}
    Consider $S' \subseteq S$ such that $K \cap S \subseteq S'$, with an $S'$-connecting system $(\mathcal{S},\mathcal{T})$. Let~$M$ be a minimum weight maximum matching of $B$ of finite weight. Then a Steiner tree $T_M$ for terminals $K$ can be constructed in $\poly(n)$ time with \[\cost(T_M) \le \omega(M) + \cost(\mathcal{T}[S']) + \sum_{k \in K} \cost( \mathsf{SP}[S', k]).\] 
\end{claim}
\begin{claimproof}
We let 
\begin{align*}
    T_M =& \mathcal{T}[S'] \cup  \{ \mathsf{MST}[G_p,S_i \cup \{k_p\}] \colon (S_i,(p,0)) \in M \} \\
    & \cup \{ \mathsf{MST}[G_p,S_i] \colon (S_i,(p,j)) \in M \text{ and } j >0 \} \\
    &\cup \{\mathsf{SP}[S',k_p] \colon k_p \in K\setminus S \text{ and } (k_p,0) \text{ not incident to }M\}.
\end{align*}

Clearly we have that $\cost(T_M) \le \omega(M) + \cost(\mathcal{T}[S']) + \sum_{k \in K} \cost( \mathsf{SP}[S', k])$. Moreover, $T_M$ can easily be computed based on $M$ in polynomial time. We now prove that $K\cup S'$ is self-reachable in $T_M$.

First we prove that every $S_i \in \mathcal{S}$ is self-reachable in $T_M$. Since $M$ is a maximum matching of finite weight and $B$ is a complete bipartite graph with $|\mathcal{S}| < |P|$, we find that all of $\mathcal{S}$ is matched in $M$ and each edge of $M$ has a finite weight. If $(S_i, (p,0)) \in M$, then all of $S_i \cup \{k_p\}$ is self-reachable in $T_M$ by $\mathsf{MST}[G_p,S_i\cup \{k_p\}]$. Else, we have $(S_i, (p,j)) \in M$ for some $j>0$, then all of $S_i$ is self-reachable in $T_M$ by $\mathsf{MST}[G_p,S_i]$. 
We can then use Lemma~\ref{lem:connectedwithT} to infer that $S'$ is self-reachable in $T_M$.  

Now we prove that each terminal $k \in K$ is reachable from some vertex in $S'$ in $T_M$. If $k \in S'$, then this is clearly true. Otherwise we have $k = k_p$ for some $p \in \{1,\dots,q\}$. If $(S_i,(p,0))\in M$ for some $S_i \in \mathcal{S}$, then $k_p$ is reachable from $S'$ through $\mathsf{MST}[G_p,S_i\cup \{k_p\}]$. Else, we have that $\mathsf{SP}[S',k_p]$ was added to $T_M$, which is a path from $k_p$ to some vertex of $S'$. Hence, we can conclude that~$K \cup S'$ is self-reachable in~$T_M$. 

Since~$T_M$ has a connected component~$T'$ containing~$S' \cup K$ and has total cost at most the claimed amount, while the weight of any edge is non-negative, any spanning tree for~$T'$ forms a Steiner tree for~$K$ whose cost is at most the claimed amount. Since we can compute a spanning tree in linear time, the claim follows.
\end{claimproof}

We are now ready to present our algorithm. For all $S' \subseteq S$ such that $S\cap K\subseteq S'$, for all $S'$-connecting systems~$(\mathcal{S}, \mathcal{T})$, construct the weighted complete bipartite graph $B$ as above. Then compute a minimum weight maximum matching $M$ of $B$. Using Claim~\ref{claim:MtoT} we construct a Steiner tree if the weight of $B$ is finite in polynomial time. Finally, the algorithm outputs the minimum-weight Steiner tree found during this process. 

The algorithm runs in $2^{\Oh(|S|\log|S|)}\poly(n)$ time and polynomial space: we consider $2^{|S|}$ different sets $S'$, for each there are at most $2^{\Oh(|S'|\log|S'|)}$ different $S'$-connecting systems by Lemma~\ref{lemma:nrScon}. Computing the weights of $B$ then takes $2^{\Oh(|S'|)}\poly(n)$ time and polynomial space using the algorithm by Fomin, Kaski, Lokshtanov, Panolan, and Saurabh~\cite{fomin2019parameterized}. Finding a minimum-cost maximum matching of $B$ takes $\poly(n)$ time and space. Constructing a Steiner tree based on $B$ takes only polynomial time and space. It remains to prove that the output of the algorithm is a minimum Steiner tree.

\begin{claim} \label{lem:polySpaceOPT}
The algorithm outputs a minimum-weight Steiner tree.
\end{claim}
\begin{claimproof}
    Given a minimum-weight Steiner tree $F$, we prove that the algorithm finds a Steiner tree of weight at most $\cost(F)$. 
    We construct an $S'$-connecting system based on the choices made by~$F$, and show that in the iteration where the algorithm considered this $S'$-connecting system, it found a Steiner tree whose cost is not larger than that of~$F$ and therefore optimal
    Let $S' \subseteq S$ be the set of vertices of $S$ used by $F$. Because $F$ is a Steiner tree for terminals $K$, we find $S\cap K \subseteq S'$. 
    
    Starting from the graph $F$, contract each connected component of $F-S$ into a single vertex. Then, remove any leaves of the resulting graph that were connected components of $F-S$ before and let the resulting graph be $\mathcal{T}$. 
    Denote the contracted vertices still in $\mathcal{T}$ by $u_1,\dots,u_m$ and let $S_i = N_\mathcal{T}(u_i)$ for all $ i \in \{1,\dots,m\}$. Finally, let $\mathcal{S} = \{S_1,\dots,S_m\}$. 

    We claim that $(\mathcal{S},\mathcal{T})$ is an $S'$-connecting system. First, note that $S_i \subseteq S'$ as each $u_i$ represents a connected component of $F-S$, hence it can only have neighbors in $S'$. Since we removed all component-representing vertices that were leaves, we have $d_{\mathcal{T}}(u_i)>1$. Finally, any edge in $\mathcal{T}[S]$ was untouched and therefore part of $F$ and of $G$. Hence, $(\mathcal{S},\mathcal{T})$ is an $S'$-connecting system.

    We will show that the algorithm will find a Steiner tree of cost at most $\cost(F)$ for this specific $S'$-connecting system. Let $B$ be the complete bipartite graph for this specific iteration of the algorithm. We will give a matching $M$ of finite weight such that $\omega(M) \le \cost(F) -\cost(\mathcal{T}[S'])  - \sum_{k \in K} \cost(\mathsf{SP}[S',k])$. Note that this is indeed sufficient, as the cost of the Steiner tree found by the algorithm is then at most $\cost(F)$ by Claim~\ref{claim:MtoT}.

    For each $u_i \in V(\mathcal{T})$, let $U_i$ be the corresponding connected component of $F$ that was contracted. Note that $U_i$ can only be part of one component $C_p$.  Let $p(i)$ denote the partition that $U_i$ is part of. Then we define $M$ as:
    \[M = \{(S_i,(p(i),0)): i \in \{1,\dots,m\},  k_p \in V(U_i)\} \cup \{(S_i,(p(i),i)): i \in \{1,\dots,m\}, k_p \not \in V(U_i)\} .\]
    
    Note that all $S_i$ are matched exactly once in this way, so $M$ is a maximum matching. We are left to analyze that $\omega(M) \le \cost(F) - \cost(\mathcal{T}[S']) - \sum_{k \in K}\cost(\mathsf{SP}[S',k])$. 
    
    Note that $F$ contains at least the following edge sets:
    \begin{itemize}
        \item $\mathcal{T}[S'] = F[S']$, i.e., all the edges between vertices of $S'$, 
        \item $U_i \cup \delta_F(U_i)$ for all $u_i$, i.e., all non-leaf connected components that were contracted, together the edges between that connected component and the vertices of $S'$,
        \item a path from $k$ to some vertex of $S'$ for all $k \in K\setminus S$ s.t. $k \not \in V(U_i)$ for all $u_i$.
    \end{itemize}
    Therefore, we have 
    \[\cost(F) \ge \cost(\mathcal{T}[S']) + \sum_{i=1}^m\cost(U_i \cup \delta_F(U_i)) + \sum_{k \in K\setminus \left(\bigcup_{i=1}^m(V(U_i)\right)} \cost(\mathsf{SP}[S',k]).\]
    Note that for all $i$ such that $U_i\cap K = \emptyset$, each subgraph $U_i \cup \delta_F(U_i)$ a is a Steiner tree for $S_i$ in $G_{p(i)}$. Therefore, we have $\cost(U_i\cup \delta_F(U_i) \ge \cost(\mathsf{MST}[G_{p(i)},S_i])$. Similarly, we find that for all $i$ such that $U_i\cap K \not = \emptyset$, we have $\cost(U_i\cup \delta_F(U_i) \ge \cost(\mathsf{MST}[G_{p(i)},S_i \cup \{k_{p(i)}\}])$. Therefore, we can conclude that:
\begin{align*}
    \cost(F) &\ge \cost(\mathcal{T}[S']) + \sum_{i\in [m]: U_i\cap K = \emptyset} \hspace{-1em}\cost(\mathsf{MST}[G_{p(i)},S_i]) \\
    &\quad + \sum_{i\in [m]: U_i\cap K \not= \emptyset}\hspace{-1em} \cost(\mathsf{MST}[G_{p(i)},S_i\cup \{k_{p(i)}\}]) + \sum_{k \in K\setminus \left(\bigcup_{i=1}^m(V(U_i)\right)}\hspace{-1em} \cost(\mathsf{SP}[S',k])\\
    & \ge \cost(\mathcal{T}[S']) + \sum_{i\in [m]: U_i\cap K = \emptyset} \hspace{-1em} \omega((S_i,(p(i),i)))\\
    &\quad + \sum_{i\in [m]: U_i\cap K \not= \emptyset}\hspace{-1em} \omega((S_i,(p(i),0))) + \sum_{k \in K} \cost(\mathsf{SP}[S',k])\\
    & = \cost(\mathcal{T}[S']) + \sum_{k \in K} \cost(\mathsf{SP}[S',k]) + \omega(M) \qedhere
\end{align*}
\end{claimproof}

This concludes the proof of \cref{thm:polyspace}.
\end{proof}

\section{Tree $K$-free-decompositions}
\label{sec:Kfreedecomp}
In this section we formally define $K$-free treewidth as the minimum width of a tree $K$-free-decomposition. We also show how to compute a 5-approximation in FPT time.

The definition of \emph{$K$-free treewidth} closely resembles the notion of tree $\hh$-decomposition that inspired it~\cite{HtreewidthEIBEN202157,JansenKW23}. The difference lies in the parts of the decomposition that do not contribute to the cost. Throughout this section, $K$ can be any set of vertices. However, we will use $K$ as the set of \textsc{Steiner Tree} terminals in the rest of this paper.

\begin{definition}[Tree $K$-free-decomposition] \label{def:tfree}
    For a graph $G$ and a set $K \subseteq V(G)$, a \emph{tree $K$-free-decomposition} of graph $G$ is a triple $(\T,\chi,L)$ where $L\subseteq V(G)$, $\T$ is a rooted tree, and $\chi: V(\T) \to 2^{V(G)}$, such that:
    \begin{enumerate}[\hspace{40pt}]
        \item[\KA] for each $v \in V(G)$, the nodes in $\{x \colon v \in \chi(x)\}$ form a non-empty connected subtree of $\T$,
        \item[\KB] for each edge $uv \in E(G)$, there is a node $x \in V(\T)$ with $\{u,v\} \subseteq \chi(x)$,
        \item[\KC] for each $v \in L$, there is a unique $x \in V(\T)$ with $v \in \chi(x)$, and $x$ is a leaf of $\T$,
        \item[\KD] for each node $x \in V(\T)$, we have $\chi(x)\cap L \cap K = \emptyset$. 
    \end{enumerate}
    The width of a tree $K$-free-decomposition $(\T,\chi,L)$ is defined as $\max\{0,\max_{x \in V(\T)} |\chi(x)\setminus L| -1\}$. The $K$-free-treewidth of a graph $G$, denoted $\twtfree(G,K)$, is the minimum width of a tree $K$-free-decomposition of $G$. 
\end{definition}

The first two items in this definition ensure that the pair~$(\T, \chi)$ forms a valid (standard) tree decomposition. The vertex set~$L$, which must be disjoint from the terminal set~$K$, corresponds to the vertices in terminal-free subgraphs that are not decomposed further. Each vertex from~$L$ occurs in exactly one bag, which is a leaf of~$\T$. The vertices from~$L$ do not contribute to the cost of the decomposition, which corresponds to the fact that the sets~$\chi(t) \cap L$ of leaf nodes~$t$ can represent arbitrarily large terminal-free subgraphs. To obtain the definition of $\mathcal{H}$-treewidth for a fixed graph class~$\mathcal{H}$, it suffices to omit~$K$ from the definition and replace condition \KD by the requirement that for each node~$x \in V(\T)$, the graph~$G[\chi(x) \cap L]$ belongs to~$\mathcal{H}$.

In order to find a tree $K$-free-decomposition of a graph $G$ of width at most $5\cdot\twtfree(G,K)+5$, we will use the algorithm by Jansen, de Kroon, and W{\l}odraczyk~\cite{JansenKW23} for finding an 5-approximation of a tree $\triangle$-free-decomposition for a graph. Recall the following definition of tree $\mathcal{H}$-free decompositions applied to the hereditary graph class of triangle-free graphs.

\begin{definition}[Tree $\triangle$-free decomposition, Definition 5 of~\cite{JansenKW23}] \label{def:deltafree}
A tree $\triangle$-free-decomposition of graph $G$ is a triple $(\T,\chi,L)$ 
where $L\subseteq V(G)$, $\T$ is a rooted tree, and $\chi: V(\T) \to 2^{V(G)}$, such that:
    \begin{enumerate}[\hspace{40pt}]
        \item[\DA] for each $v \in V(G)$, the nodes in $\{x \colon v \in \chi(x)\}$ form a non-empty connected subtree of $\T$,
        \item[\DB] for each edge $uv \in E(G)$, there is a node $x \in V(\T)$ with $\{u,v\} \subseteq \chi(x)$,
        \item[\DC] for each $v \in L$, there is a unique $x \in V(\T)$ with $v \in \chi(x)$, and $x$ is a leaf of $\T$,
        \item[\DD] for each node $x \in V(\T)$, the graph $G[\chi(x)\cap L]$ contains no triangles.
    \end{enumerate}
    The width of a tree $\triangle$-free-decomposition $(\T,\chi,L)$ is defined as $\max(0,\max_{x \in V(\T)} |\chi(x)\setminus L| -1)$. The $\triangle$-free-treewidth of a graph $G$, denoted $\twdeltafree(G)$, is the minimum width of a tree  $\triangle$-free-decomposition of $G$. 
\end{definition}

Since the class of $\triangle$-free graphs is defined by a finite family of forbidden subgraphs (namely, that of a triangle), the problem of $\triangle$-free-deletion can be solved in $3^s\cdot \poly(n)$ time, where $s$ is the size of the solution \cite{vertex-deletion} by a standard bounded-depth branching algorithm. This directly implies the following result. 

\begin{proposition}[Jansen, de Kroon, W{\l}odarczyk, \cite{JansenKW23}]\label{prop:deltafree}
    There is an algorithm that, given a graph $G$ and integer $k$, either computes a tree $\triangle$-free-decomposition of width at most $5k+5$ consisting of $\Oh(n)$ nodes, or correctly concludes that $\twdeltafree(G)>k$. The algorithm runs in time $2^{\Oh(k)} \cdot \poly(n)$ and polynomial space.
\end{proposition}

Using these ingredients, we present an FPT 5-approximation for $K$-free treewidth.

\begin{theorem}
    \label{thm:findKfreedec}
    There is an algorithm that takes an input graph $G$, vertex set $K\subseteq V(G)$, integer $k$, and either computes a tree $K$-free-decomposition of width at most $5k+5$ consisting of $\Oh(n)$ nodes, or correctly concludes that $\twtfree(G,K)>k$. The algorithm runs in time $2^{\Oh(k)} \cdot \poly(n)$ and polynomial space.
\end{theorem}
\begin{proof}
    We will first construct a graph $\hat G$ from $G$ and use Proposition~\ref{prop:deltafree} on $\hat G$. We will then use the resulting tree $\triangle$-free-decomposition of $\hat{G}$ to construct a tree $K$-free-decomposition for $G$ of width at most $5k+5$. 

    To get $\hat{G}$, we will subdivide every edge of $G$ and replace every vertex of $K$ with a triangle. Formally, we define $\hat{G}$ as follows. Let \[V(\hat{G}) = V(G) \cup \{ e \colon e\in E(G) \}  \cup  \{ k', k'' \colon k \in K\}, \text{ and}\]
    \[E(\hat{G}) = \left \{ \{u, e\} \colon u \in V(G), e\in E(G), u \in e\right\} \cup \left \{\{k, k'\}, \{k,k''\}, \{k',k''\} \colon k \in K\} \right \} .\] 

    An important property of $\hat{G}$ that follows from this construction, is that the only triangles of $\hat{G}$ are of the form $\{k,k',k''\}$ for $k\in K$. We first show that the $\triangle$-free-treewidth of $\hat{G}$ is at most the $K$-free-treewidth of $G$. 

    \begin{claim} For any graph $G$ we have 
        $\twdeltafree(\hat{G}) \le \max\{\twtfree(G,K) ,1\}$. 
    \end{claim}
    \begin{claimproof}
        Take $(\T, \chi,L)$ as a tree $K$-free-decomposition for $G$ of minimum width. We will create a tree $\triangle$-free-decomposition $(\hat\T,\hat\chi, \hat L)$ for $\hat G$ with the same or smaller width.
        Take $(\T,\chi, L)$ as a basis for $(\hat\T,\hat\chi,\hat L)$, but we extend it in the following way:
        \begin{itemize}
            \item For any $k \in K$ add $k', k''$ to $\hat L$. Take any $x \in V(\T)$ such that $k \in \chi(x)$. If $x$ is a leaf, add $k'$ and $k''$ to $\hat\chi(x)$. If $x$ is not a leaf, add a new node $x_k$ to $V(\hat \T)$ as a child of $x$ with $\hat\chi(x_k) = \{k,k',k''\}$.
            \item For any $\{u,v\} = e \in V(\hat G)$ add $e$ to $\hat L$. Take any $x \in V(\T)$ such that $e\subseteq \chi(x)$. If $x$ is a leaf, add $e$ to $\hat\chi(x)$. If $x$ is not a leaf, add a new node $x_e$ to $V(\hat \T)$ as a child of $x$ with $\hat\chi(x_{e}) = \{u,v,e\}$.
        \end{itemize}

        First note that for all original nodes $x \in V(\T)$ we have $|\chi(x) \setminus L| \le |\hat \chi(x) \setminus \hat L|$, as whenever vertices were added to $\hat\chi(x)$, they were also added to $\hat L$. Moreover, for all new nodes $x \in V(\hat\T) \setminus V(\T)$, we have $|\hat\chi(x)\setminus L| \le 2$. Hence the width of $(\hat\T,\hat\chi, \hat L)$ is upper bounded by the maximum of $1$ and the width of $(\T,\chi,L)$, i.e., $\twtfree(G,K)$.
        
        We will prove that $(\hat\T,\hat\chi, \hat L)$ has all the properties as in Definition~\ref{def:deltafree}. Note that \DA, \DB, and \DC clearly hold by our construction. For \DD, we first note that $\hat G$ only contains triangles of the form $\{k,k',k''\}$, as we subdivided $G$ at every edge to get $\hat G$. Note that for any $k \in K$ we have $k \not \in L$: otherwise $k$ would have been part of a leaf $x$ (by \KC) and then $\chi(x)$ would contain a vertex of $K \cap L$, which would contradict \KD. Therefore, as $k \not\in \hat L$ for all $k \in K$, we have that $G[\hat L]$ contains no triangles. 
    \end{claimproof}

    Hence, if we apply Proposition~\ref{prop:deltafree} to $\hat G$, we will either get the output that $5k+5 < \twdeltafree(\hat G) \le \twtfree(G,K)$, or we get a tree $\triangle$-free-decomposition $(\hat \T, \hat \chi, \hat L)$ of width at most $5k+5$ for $\hat G$. We will show how to use $(\hat \T, \hat \chi, \hat L)$ as a basis for a tree $K$-free-decomposition for $G$. 
    
    \begin{claim}
        Given tree $\triangle$-free-decomposition $(\hat \T, \hat \chi, \hat L)$ for $\hat G$, we can find in polynomial time a tree $K$-free decomposition for $G$ of the same or smaller width. 
    \end{claim}
    \begin{claimproof}
        First we show that we can assume that $k \not \in \hat{L}$ for any $k \in K$. Indeed, if $k \in \hat{L}$, then there is a unique leaf $x \in V(\T)$ such that $k \in \hat\chi(x)$. Since $k'$ and $k''$ are adjacent to $k$ in $\hat G$, we find $\{k,k',k''\} \subseteq \chi(x)$. Moreover, as $k'$ only has $k$ and $k''$ as neighbors, we may assume that $x$ is the only bag containing $k'$ (similarly, $k''$ is only contained in $\chi(x)$). Note that not all of $\{k,k',k''\}$ can be part of $\hat L$ as it is a triangle, assume that $k' \not \in \hat L$. Then change $\hat L$ such that $k' \in \hat L$ and $k \not \in \hat L$. This creates a tree $\triangle$-free-decomposition for $\hat G$ of the same width as $(\hat \T, \hat \chi, \hat L)$. Hence, we may assume that $\hat L \cap K = \emptyset$.
    
        We create a tree $K$-free decomposition $(\T,\chi,L)$ for $G$ based on $(\hat \T, \hat \chi, \hat L)$. For every $e \in E(G) \setminus$, we choose an arbitrary endpoint $u$ and let $\alpha(e) = u$. This $u$ will replace $e$ in the new decomposition. For this to work, We have to adjust $L$ a bit to take into account the case where $\alpha(e) \in \hat L$, but $e \not \in \hat L$. Namely, then suddenly $\alpha(e)$ occurs in non-leaf vertices so it should not be part if $L$. We create the following tree $K$-free-decomposition for $G$:

        \begin{itemize}
            \item $\T = \hat\T$,
            \item $L = (\hat L \cap V(G) )  \setminus \{ \alpha(e) \colon e \in E(G) \setminus \hat L\}$,
            \item For all $x \in V(\T)$ we take  \[\chi(x) = (\hat \chi(x) \cap V(G)) \cup \{ \alpha(e) \colon e\in E(G)\cap \hat\chi (x)\}.\]
        \end{itemize}

        We show that $|\chi(x)\setminus L| \le |\hat \chi(x)\setminus \hat L|$ for all $x \in V(\T)$. 
        Take $x \in V(\T)$. We will show that for any $v \in \chi(x)\setminus L$, either there exists an edge $e \in \hat\chi(x)\setminus \hat L$ with $v = \alpha(e)$, or we have $v \in \hat\chi(x) \setminus \hat L$. Since for $v \in \chi(x)\setminus L$, the sets  $\{ e: v = \alpha(e) \} \cup \{v\}$ for different $v$ are disjoint, this will then prove   
        $|\chi(x)\setminus L| \le |\hat \chi(x)\setminus \hat L|$.
        So take any $v \in \chi(x)\setminus L$. Note that when $v = \alpha(e)$ for some $e \in \hat \chi(x) \setminus \hat L$, we directly find what we wanted to show. Else (i.e., if $v \not = \alpha(e)$ for all $e \in \hat \chi(x)\setminus \hat L$), by definition of $\chi(x)$, we have $v \in V(G)\cap \hat \chi(x)$. By definition of $L$, we then find that $v \not \in \hat L$. In that case, we directly find $v \in \hat\chi(x)\setminus \hat L$, i.e., what we were set out to prove.  
        
        We will now prove that this decomposition has all the required properties. For \KA consider any $v \in V(G)$.  Note that $\chi^{-1}(v)$ is the union of connected subtrees, namely that of $\hat\chi^{-1}(v)$ and that of $\hat\chi^{-1}(e)$ for each $e\in E(G)$ with $\alpha(e) = u$. All of these subtrees overlap with $\hat\chi^{-1}(u)$, hence \KA holds. 
        
        For \KB, consider any edge $e = \{u,v\} \in E(G)$ 
        and assume $u = \alpha(e)$. Then \DB and \DD for 
        $(\hat \T,\hat \chi, \hat L)$ tell us that there is a $x \in V(\T)$ such that $\{e,v\} \subseteq \hat \chi(x)$. By definition of $\chi(x)$, this means that $\{\alpha(e),v\} = \{u,v\} \subseteq \chi(x)$.

        For \KC, take any $v \in L$ and assume for the sake of contradiction that the property does not hold. This means that there is a non-leaf $x \in V(\T)$ such that $v \in \chi(x)$. Since $v \in L$, we also have $v \in \hat L$. Hence $v \not \in \hat \chi (x)$ as $x$ is not a leaf vertex. This would imply that $v = \alpha(e)$ for some $e \in E(G) \cap \hat \chi(x)$. Since $x$ is a non-leaf vertex we find $e \in E(G)\setminus \hat L$. But then $\alpha(e) = v$ should not be part of $L$, i.e., we find a contradiction.
        
        For \KD, we simply note that we were able to assume that $k \not \in \hat L$ for all $k\in K$ at the beginning of this proof. 
    \end{claimproof}

This concludes the proof of \cref{thm:findKfreedec}.
\end{proof}


\subsection{Nice tree $K$-free-decompositions} \label{sec:nice:dec}
Many algorithms using standard tree decompositions use \emph{nice} tree decompositions (e.g.,~\cite{cygan2015parameterized}). We will use a very similar notion of a \emph{nice tree $K$-free-decomposition}, where every non-leaf node of a rooted tree decomposition $(\T,\chi)$ is one of five types:
\begin{itemize}
    \item \textbf{join node:} a node $x$ with exactly two child nodes $c_1$, $c_2$, with $\chi(x) = \chi(c_1) = \chi(c_2)$.
    
    \item \textbf{vertex introduce node:} a node $x$ that \emph{introduces a vertex} $v \in V(G)$ has exactly one child node $c$, such that $\chi(x) = \chi(c)\cup \{v\}$ with $v \not \in \chi(c)$. 
    
    \item \textbf{vertex forget node:} a node $x$ that \emph{forgets a vertex} $v \in V(G)$ has exactly one child node $c$, such that $\chi(x) = \chi(c)\setminus \{v\}$ with $v \in \chi(c)$. 
    
    \item \textbf{edge introduce node:} a node $x$ that \emph{introduces an edge} $\{u,v\}$ has exactly one child node $c$, such that $\chi(x) = \chi(c)$ and $\{u,v\}\subseteq \chi(x)$. 

    \item \textbf{leaf introduce node:} a node $x$ that has exactly one child node $c$ that is a leaf, such that $\chi(x) = \chi(c)\setminus L$.  
\end{itemize} 
We say that in a leaf node $c$, all edges of $G[\chi(c) \setminus L]$ are introduced. In a leaf introduce node $x$ with child node $c$, all remaining edges (i.e., all edges of $G[\chi(x)]$ except those in $G[\chi(c)\setminus L]$ are introduced.  We require that every edge is introduced exactly once in the whole decomposition. We use the notation $G_x$ to denote the subgraph of $G$ with $V(G_x)$ the union of all bags of descendants of $x$ in $\T$ and $E(G_x)$ the edges that are introduced by node $x$ or a descendant of $x$. 

Note that we can transform any tree $K$-free-decomposition into a \emph{nice} tree $K$-free-decomposition, in almost exactly the same as transforming a standard tree decomposition into a nice tree decomposition (cf.~\cite[Lemma 7.4]{cygan2015parameterized}, \cite{kloks1994treewidth}). First we apply the usual transformations on the interior of $\T$, i.e.we ensure that $\T$ is binary and introduce / forget at most one vertex per bag. Note that we only introduce edges that are not part of $G[\chi(x)]$ for a leaf $x \in V(\T)$. Then, we will do the following transformation for each leaf $c \in V(\T)$: Remove the edge between $c$ and its parent $p$ and add a node $x$ as a child of $p$ and a parent of $c$. Set $\chi(x) = \chi(c)\setminus L$. Now $x$ is a leaf introduce node, introducing the leaf $c$. This ensures all the properties of a nice tree decomposition without increasing the width of the decomposition. 

\section{Single-exponential time algorithm parameterized by $K$-free treewidth}
\label{sec:singleexpalgorithm}
In order to obtain a single-exponential algorithm even though the number of partitions of a size-$k$ set is super-exponential in~$k$, we use the rank-based approach by Bodlaender, Cygan, Kratsch, and Nederlof~\cite{BODLAENDER201586}. The main idea behind this approach is to derive single-exponential rank bounds for matrices that encode which pairs of `partial solutions' combine into a full solution. These bounds imply that any set of partial solutions can be trimmed down to a \emph{representative} subset of single-exponential size while providing the following guarantee: if the original set contained a partial solution that can be extended to an optimal full solution, then the representative set still contains such a partial solution.

We start making these ideas concrete for \textsc{Steiner Tree}. Given a node $x$ of a tree $K$-free decomposition~$(\T, \chi, L)$ of the input graph~$G$ obtained via \cref{thm:findKfreedec}, we will work bottom-up in the decomposition tree. For each node~$x \in V(\T)$ we compute a set of partial solutions in the associated graph~$G_x$ defined in \cref{sec:nice:dec}, using the previously computed data for the children of~$x$. During this process, we decide which partial solutions in $G_x$ to keep based on their behaviour on the boundary, i.e., on $\chi(x)$. Roughly speaking, a partial solution corresponds to a subgraph~$F$ of~$G_x$ containing all vertices of~$K \cap V(G_x)$. The subgraph~$F$ does not have to be connected, but if it is disconnected then each connected component contains a vertex of~$\chi(x)$. The behavior of~$F$ can be summarized by a partition of the vertices of~$\chi(x) \cap V(F)$ based on their spread over connected components of~$F$: two vertices $u,v \in \chi(x) \cap V(F)$ are in the same set of the partition if and only if they are in a common connected component of~$F$. To work with representative sets of partial solutions, we therefore need some terminology to work with partitions. We view a partition of a set~$U$ as a set of disjoint subsets of~$U$ whose union is~$U$.

\begin{definition}
For a finite set~$U$ we use~$\Pi(U)$ to denote the set of all partitions of~$U$. For $P, Q \in \Pi(U)$ we let $P\sqsubseteq Q$ if and only if for all $A \in Q$ there exists a set $A' \in P$ such that $A \subseteq A'$. If this holds, we say that $Q$ is a \emph{refinement} of $P$. 

For~$P,Q \in \Pi(U)$ we denote by~$P \sqcup Q$ the \emph{join} of partitions~$P$ and~$Q$ in the partition lattice, which can be obtained as follows. Let~$G_P$ (resp.~$G_Q$) be a graph on vertex set~$U$ whose connected components partition~$U$ according to~$P$ (resp.~$Q$). Then~$P \sqcup Q$ corresponds to the partition of~$U$ formed by the connected components of the union~$G_P \cup G_Q$. This implies that for every $A \in P \cup Q$ there exists $A' \in P\sqcup Q$ such that $A \subseteq A'$.
\end{definition}

Using this notation, we can introduce the concept of a set~$\mathcal{R}$ of partial solutions \emph{representing} another~$\mathcal{A}$. To facilitate the discussion of our algorithm, we will work with partial solutions of two types. The most intuitive form consists of subgraphs~$F$ of~$G_x$ as described above. A more succinct representation of the essential information can be obtained by encoding each such subgraph~$F$ as a pair~$(P,w)$, where~$P$ is the partition it induces on the boundary vertices and~$w$ is the integer giving the total cost of the edges in~$F$. From this perspective, the objects stored in a set of partial solutions consist of pairs~$(P,w)$. The following definition captures the idea of two partial solutions together forming a full solution: they merge together to provide a subgraph consisting of a single connected component whenever the partition representing the combined connectivity information consists of a single set containing the entire boundary.

\begin{definition}[Definition 3.4 from \cite{BODLAENDER201586}]\label{def:represents:old}
    For a universe $U$ and two sets $\mathcal{A}, \mathcal{R} \subseteq \{(P,w): P \in \Pi(U) \wedge w \in \nat \}$, we say that $\mathcal{R}$ \emph{represents} $\mathcal{A}$ if for all $Q \in \Pi(U)$ we have
    \[\min\{w \colon (P,w)\in\mathcal{A} \wedge P\sqcup Q = \{U\}\}  = \min\{w \colon (P,w)\in\mathcal{R} \wedge P\sqcup Q = \{U\}\} .\]
\end{definition}

Bodlaender et al.~\cite{BODLAENDER201586} showed that a representative set $\mathcal{R}$ can be computed by finding a basis of a suitable matrix of rank~$2^{|U|-1}$. We can therefore bound the size of $\mathcal{R}$ by $2^{|U|-1}$. 

\begin{proposition}[Theorem 3.7 from \cite{BODLAENDER201586}] \label{prop:findrepset} 
There is an algorithm that given a universe $U$ and a set $\mathcal{A} \subseteq \{(P,w): P \in \Pi(U) \wedge w \in \nat \}$ in time $2^{\Oh(|U|)} \poly(|\mathcal{A}|)$ finds a set $\mathcal{R} \subseteq \mathcal{A}$ of size at most $2^{|U|-1}$, that represents $\mathcal{A}$. 
\end{proposition}

We slightly adjust this proposition to also work for storing subgraphs (i.e., partial solutions) as representative sets, instead of partitions. This facilitates the presentation of the remainder of the algorithm. For this purpose, we define how to map subgraphs to partitions.

\begin{definition} \label{def:represent:partition}
    For a graph $G$, a set $X \subseteq V(G)$, and $F\subseteq G$ a subgraph of $G$, let $\pi_F(X)$ be the partition of $X$ such that: 
    \begin{itemize}
        \item any $u \in X\setminus V(F)$, we have $\{u\} \in \pi_F(X)$,
        \item for any $u, v \in X \cap V(F)$, vertices~$u$ and $v$ are in the same set of $\pi_F(X)$ if and only if $u$ and $v$ belong to a common connected component of $F$.
    \end{itemize} 
\end{definition}

See \cref{fig:partition} for an example of a subgraph $F$ and the partition $\pi_F(X)$. Observe that if~$C$ is a connected component of~$F$ and~$H = F \setminus C$, then~$\pi_F(X) = \pi_H(X) \sqcup \pi_C(X)$; we will use this fact later. Next, we adjust the definition of representing to work for subgraphs. We use~$\{ F \subseteq G \}$ to denote the set of all subgraphs of a graph~$G$.

\begin{figure}[t]
    \centering
    \begin{tikzpicture} [scale = 1, 
    vertex/.style = {circle, draw, fill = black, align=center,minimum size = 3pt, inner sep = 0},
    terminal/.style = {rectangle, draw, fill = black, align=center,minimum size = 5pt, inner sep = 0}
,label distance=-2pt]

\tikzstyle{every path}=[line width=1pt]


\fill[fill = gray!80] (3.5,0) ellipse (4.5 and .9);


\node[] at  (8.5,0) {$X$};

\node[vertex,label=above:$x_1$] (x1) at  (0,0) {};
\node[vertex,label=above:$x_2$] (x2) at  (1,0) {};
\node[vertex,label=above:$x_3$] (x3) at  (2,0) {};
\node[vertex,label=above:$x_4$] (x4) at  (3,0) {};
\node[vertex,label=above:$x_5$] (x5) at  (4,0) {};
\node[vertex,label=above:$x_6$] (x6) at  (5,0) {};
\node[vertex,label=above:$x_7$] (x7) at  (6,0) {};
\node[vertex,label=above:$x_8$] (x8) at  (7,0) {};

\draw (x1)  -- (-1,-1) -- (-0.5,-2) -- (.4,-1.5);
\draw (x2) -- (1.5,-1) -- (1.4, -1.8) -- (2.5,-1) -- (x3);
\draw (x4) -- (2.5,-1);
\draw (x4) -- (3.5,-1.5) -- (4.2,-1.3) -- (x5);
\draw (3.5,-1.5) -- (3.2,-1.7) -- (3.9,-2);
\draw (x6) -- (5.2,-1) -- (4.7,-1.5) --(5.2,-2) -- (5.9,-1.3)-- (x7);
\draw (6.4,-1.4)-- (x7);

\end{tikzpicture}
    \caption{In this example we have $X=
    \{x_1,\dots,x_8\}$. The lines are a visual representation of a subgraph $F$ of $G$, the graph $G$ itself is not drawn. For this subgraph $F$ we have $\pi_F(X) = \{\{x_1\},\{x_2, x_3, x_4, x_5\},\{x_6,x_7\},\{x_8\}\}$.}
    \label{fig:partition}
\end{figure}
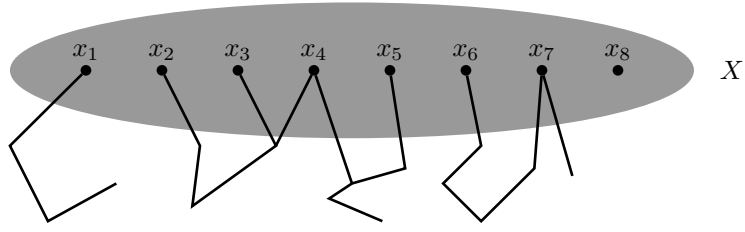

\begin{definition} \label{def:represent:subgraph}
    Given an edge-weighted graph $G$, a vertex set $Z \subseteq V(G)$, and two sets of subgraphs $\mathcal{B}, \mathcal{R} \subseteq \{F\subseteq G \}$, we say that $\mathcal{R}$ \emph{represents} $\mathcal{B}$ \emph{on $Z$} if for all $Q \in \Pi(Z)$ we have
    \begin{equation}
      \min\{\cost(F) \colon F \in \mathcal{B} \wedge \pi_F(Z) \sqcup Q = \{Z\}\}  = \min\{\cost(F) \colon F \in \mathcal{R}  \wedge \pi_F(Z)\sqcup Q = \{Z\}\}. \label{eq:defrep}  
    \end{equation}
\end{definition}

The given definitions support the following straight-forward extension of \cref{prop:findrepset}, which facilitates computing representative subsets of partial solutions stored as subgraphs. The computation works with respect to \emph{any} choice of vertex set~$Z$ as boundary. The fact that~$Z$ is not necessarily a separator in~$G$ is crucial for its later use. 

\begin{corollary} \label{cor:repsets} Consider an edge-weighted graph $G$ and a subset of the vertices $Z \subseteq V(G)$. For a given set of subgraphs $\mathcal{B} \subseteq \{ F \subseteq G\}$, we can find in $2^{\Oh(|Z|)} \poly(|\mathcal{B}|)$ time a set $\mathcal{R} \subseteq \mathcal{B}$ of size at most $2^{|Z|-1}$ that represents $\mathcal{B}$ on $Z$.
\end{corollary}
\begin{proof}
    Let $\mathcal{B}' = \{(\pi_F(Z), \cost(F)): F \in \mathcal{B} \}$. 
    Now we can apply Proposition~\ref{prop:findrepset} to $\mathcal{B}'$ to find a set $\mathcal{R}' \subseteq \mathcal{B}'$ of size at most $2^{|Z|-1}$ in time $2^{\Oh(|Z|)} \poly(|\mathcal{B}|)$, that represents $\mathcal{B}$ on universe $Z$. 

    Create the set $\mathcal{R}$ by adding for each $(P,w) \in \mathcal{R}'$ exactly one subgraph $F$ from $\mathcal{B}$ with $\pi_F(Z) = P$ and $\cost(F)=w$ (these must exist by definition of $\mathcal{B}'$ and the fact that $\mathcal{R}'\subseteq \mathcal{B}'$). Note, that $|\mathcal{R}| \le 2^{|Z|-1}$ and $\mathcal{R}\subseteq \mathcal{B}$. Hence, we are left to show that $\mathcal{R}$ represents $\mathcal{B}$,  
    i.e., for all $Q \in \Pi(Z)$ that \eqref{eq:defrep} holds. 

Fix some $Q \in \Pi(Z)$.
Note that the left-hand side is at most the right-hand side as $\mathcal{R}\subseteq \mathcal{B}$.
For the other direction, consider $F \in \mathcal{B}$ that minimizes the left side of the equation. We know $(\pi_F(Z),\cost(F)) \in \mathcal{B}'$ and because $\mathcal{R}'$ represents $\mathcal{B}'$ on universe $Z$, we find that there exists $(P, w) \in \mathcal{R}'$ such that $w = \cost(F)$ and $P \sqcup Q = \{Z\}$. The way we created $\mathcal{R}$ implies that there exists a $F'\in \mathcal{R}$ such that $\pi_{F'}(Z) \sqcup Q = \{Z\}$ and $\cost(F') = \cost(F)$. Hence, this shows the equality. 
\end{proof}

We now present the main ingredient for applying the rank-based approach for $K$-free treewidth. It will be used to compute the table entries of the leaf-introduce vertices of the tree $K$-free decomposition.

\begin{lemma}\label{lem:leafcase}
    Given as input an edge-weighted $n$-vertex graph~$G$ and two disjoint subsets of vertices $Y\subseteq V(G)$ and $Z \subseteq V(G)$ with~$Z \neq \emptyset$, we can compute in $2^{\Oh(|Z|)} \cdot \poly(n)$ time a set $\mathcal{R}(Z) \subseteq \mathcal{F}(Z) = \left \{F \subseteq G[Z\cup Y] \right  \}$ such that:
    \begin{itemize}
    \item $|\mathcal{R}(Z)| \le 2^{|Z|-1}$, and
    \item $\mathcal{R}(Z)$ represents $\mathcal{F}(Z)$ on $Z$.
    \end{itemize}
\end{lemma}

\begin{proof}
We want to find a representative set for all possible subgraphs in $G[Y \cup Z]$. For this purpose, we build $\mathcal{F}(Z)$ by increasing the size of $Z$, while maintaining the bounded size of $\mathcal{F}(Z)$. Order the elements of $Z$ arbitrarily, i.e., let $Z = \{z_1,\dots,z_{k}\}$. Let $Z_i = \{z_1,\dots,z_i\}$ be the first $i$ elements of $Z$ and define $\mathcal{F}(Z_i) = \{ F \subseteq G[Z_i \cup Y] \}$.
We will iteratively compute sets $\mathcal{R}(Z_i) \subseteq \mathcal{F}(Z_i)$ with the following properties:
\begin{itemize}
    \item $|\mathcal{R}(Z_i)| \le 2^{|Z|-1}$, and
    \item $\mathcal{R}(Z_i)$ represents $\mathcal{F}(Z_i)$ on $Z$.
    \end{itemize}

We set $\mathcal{R}(\emptyset) = \{ \emptyset \}$.
To compute $\mathcal{R}(Z_i)$ from $\mathcal{R}(Z_{i-1})$ we execute the following steps.
\begin{enumerate}
\item Initialize $\mathcal{B}(Z_i) = \mathcal{R}(Z_{i-1})$. 

\item For all $F\in \mathcal{R}(Z_{i-1})$, for all $T \subseteq Z_{i}$ such that $z_i \in T$, compute a minimum-cost Steiner tree $F_T$ for terminals $T$ in $G[Z_i\cup Y]$ in $2^{\Oh(|Z|)}\poly(n)$ time using the Dreyfus-Wagner algorithm~\cite{dreyfus1971steiner}.
Add $F\cup F_T$ to $\mathcal{B}(Z_i)$. 

\item Compute a representative set $\mathcal{R}(Z_{i}) \subseteq \mathcal{B}(Z_{i})$ with $|\mathcal{R}(Z_i)| \le 2^{|Z|-1}$ using Corollary~\ref{cor:repsets}.
\end{enumerate}

The resulting set~$\mathcal{R}(Z_k) = \mathcal{R}(Z)$ is given as the output. The main work happens in Step~2. For each of the~$|\mathcal{R}(Z_{i-1})| \leq 2^{|Z|-1}$ choices for~$F$, there are at most~$2^{|Z|}$ choices for~$T$. Hence we invoke the Dreyfus-Wagner algorithm~$2^{\Oh(|Z|)}$ times for a terminal set of size~$|T| \leq |Z|$. The resulting set~$\mathcal{B}(Z_i)$ has size~$2^{\Oh(|Z|)}$, so that \cref{cor:repsets} also runs in time~$2^{\Oh(|Z|)} \poly(n)$ for each of the~$k + 1 \in \Oh(n)$ choices of~$i$. The run-time bound follows.

We prove correctness by induction over $i$. 
For the base case~$i=0$, we argue that~$\mathcal{R}(\emptyset)$ satisfies the two conditions. The set $\mathcal{F}(Z_0) = \mathcal{F}(\emptyset)$ only contains subgraphs $F$ of $G[Y]$ and for any of these we have $\pi_F(Z) = \{\{z_1\},\{z_2\},\dots,\{z_k\}\}$ since~$Y \cap Z = \emptyset$. Hence, the only $Q \in \Pi(Z)$ satisfying $\pi_F(Z) \sqcup Q = \{Z\}$ is $Q = \{\{Z\}\}$. Since for $F=\emptyset$ we have $\pi_F(Z) \sqcup Q = \{Z\}$, we find that $\mathcal{R}(\emptyset) = \{\emptyset\}$ is a representative set for $\mathcal{F}(\emptyset)$ on $Z$.     

For the case $i>0$, first note that any time we add a subgraph~$F \cup F_T$ to~$\mathcal{B}(Z_i)$ in Step 2, we have~$F \in \mathcal{R}(Z_{i-1}) \subseteq \mathcal{F}(Z_{i-1}) \subseteq \mathcal{F}(Z_i)$, and~$F_T \subseteq G[Z_i \cup Y]$. Hence their union is a subgraph of~$G[Z_i \cup Y]$. Therefore, $\mathcal{R}_i(Z) \subseteq \mathcal{B}_i(Z) \subseteq \mathcal{F}_i(Z)$. 
Moreover, we have $|\mathcal{R}(Z_i)|\le 2^{|Z|-1}$ because we constructed it using \cref{cor:repsets}. All that remains to show is that $\mathcal{R}(Z_i)$ represents $\mathcal{F}(Z_i)$ on $Z$, i.e., for all $Q \in \Pi(Z)$ we have 
\begin{equation*}
\min\{\cost(F) \colon F \in \mathcal{F}(Z_i) \wedge \pi_F(Z) \sqcup Q = \{Z\}\} = \min\{\cost(F) \colon F \in \mathcal{R}(Z_i)  \wedge \pi_F(Z)\sqcup Q = \{Z\}\}.
\end{equation*}

As $\mathcal{R}(Z_i) \subseteq \mathcal{F}(Z_i)$ we directly see that the left-hand side is at most the right-hand side. For the other direction, fix $Q \in \Pi(Z)$ and choose $F \in \mathcal{F}(Z_i)$ that minimizes the left-hand side. Let $C$ be the connected component of $F$ containing~$z_i$ (possibly $C = \emptyset$). Let $H = F \setminus C$. Note that $H \in \mathcal{F}(Z_{i-1})$ as now $H \subseteq G[Z_{i-1} \cup Y]$. Take $Q' = \pi_C(Z) \sqcup Q$. We find that
\[\pi_H(Z) \sqcup Q' = \pi_H(Z) \sqcup \pi_C(Z)\sqcup Q  = \pi_F(Z) \sqcup Q = \{Z\}.\] 

Using induction, we find that there exists $H' \in \mathcal{R}(Z_{i-1})$ such that $\pi_{H'}(Z) \sqcup Q' = \{Z\}$ and $\cost(H') = \cost(H) = \cost(F)- \cost(C)$. The algorithm at one point considered $H'$ and $T = V(C) \cap Z$ in Step 2, and added $ H' \cup F_T$ to $\mathcal{B}(Z_i)$. Note that 
\begin{align*}\pi_{H' \cup F_T}(Z) \sqcup Q & = \pi_{H' \cup F_T}(Z) \sqcup Q \sqcup \pi_C(Z) && \text{ as } F_T \text{ connects $V(C)\cap Z$ }\\
& = \pi_{H' \cup F_T}(Z) \sqcup Q' && \text{ by definition of~$Q'$} \\
& = \{Z\}.&& \text{ as } \pi_{H'}(Z) \sqcup Q' = \{Z\}
\end{align*}

Using Proposition~\ref{prop:findrepset}, we find that there exists a subgraph $ F' \in \mathcal{R}(Z_{i})$ such that $F' \sqcup Q = \{Z\}$. To bound its cost, note that $\cost(F_T) \le \cost(C)$ as $F_T$ is a minimum-cost Steiner tree in $G[Z_i \cup Y]$ for terminal set $T$ and~$C$ is a candidate solution. Then we derive
\begin{align*} 
\cost(F') &= \cost(H' \cup F_T) 
\le \cost(H') + \cost(F_T) 
= \cost(F) - \cost(C) + \cost(F_T) \\
&\le \cost(F) - \cost(C) + \cost(C) 
= \cost(F).
\end{align*}
Hence, $F'$ shows that the left-hand side of the equation is greater than or equal to the right-hand side. This concludes the proof of \cref{lem:leafcase}.
\end{proof}

We can use Lemma~\ref{lem:leafcase} to solve \textsc{Steiner Tree} in $2^{\twtfree(G,K)} \poly(n)$ time. Below, we give a self-contained presentation of the existing~\cite{BODLAENDER201586} rank-based dynamic-programming algorithm for \textsc{Steiner Tree} parameterized by treewidth---thereby correcting some small issues from the literature---and show how \cref{lem:leafcase} allows it to be extended for $K$-free treewidth.

\terminaltwidth*
\begin{proof}
    The first step is to use Theorem~\ref{thm:findKfreedec} to find a nice tree $K$-free-treewidth decomposition $(\T,\chi,L)$ of width $\le 5\twtfree(G,K)+5$.

    For the actual algorithm on the decomposition, we use the algorithm from Bodlaender, Cygan, Kratsch, and Nederlof~\cite{BODLAENDER201586} to solve weighted Steiner tree on graphs of bounded treewidth as a building block. Let us recall what this algorithm computes. Recall that for any node $x \in V(\T)$, the graph $G_x$ is defined as the subgraph of $G$ with as vertices any vertex that was introduced in (a descendant of) $x$ and as edges any edge that was introduced in (a descendant of) $x$. The algorithm from Bodlaender, Cygan, Kratsch, and Nederlof~\cite{BODLAENDER201586} then computes, for every node $x \in V(\T)$ and $Z \subseteq \chi(x)$, a set $\mathcal{A}'_x(Z)$ that represents $\mathcal{A}_x(Z)$ (according to Definition~\ref{def:represent:partition}) where\footnote{There are two major differences between the definition of table entries in Section 3.3 of the paper by Bodlaender, Cygan, Kratsch, and Nederlof~\cite{BODLAENDER201586} and the one given here. The first is the inclusion that all terminals in $G_x$ should be in $V(F)$, which was simply missing before. The second is that $V(F)\cap \chi(x) \subseteq Z$ was changed to equality, to correctly ensure that each connected component of $F$ has at least one endpoint in $Z$.\label{foot:issues}}  
    \[\mathcal{A}_x(Z)  = \left\{\left(P, \min_{F \in \mathcal{E}_x(P,Z)} \cost(F)\right) \colon P \in \Pi(Z) \wedge \mathcal{E}_x(P,Z) \not = \emptyset\right\},\]
    and
    \[\mathcal{E}_x(P,Z) = \left\{ F \subseteq G_x \colon  \begin{aligned}&
    V(F) \cap \chi(x) = Z, \pi_F(Z)  = P, \\ &\#\mathsf{cc}(F) = |P|, V(G_x)\cap K \subseteq V(F)  \end{aligned}\right\}.\]

Intuitively,~$\mathcal{E}_x(P,Z)$ is the set of all subgraphs~$F$ of~$G_x$ whose connected components partition the vertices of~$Z$ exactly as prescribed by~$P$, which contain all terminals of~$G_x$, for which every connected component intersects~$\chi(x)$, and for which~$Z$ coincides with the set of vertices from bag~$\chi(x)$ used by~$F$. The set~$\mathcal{A}_x(Z)$ then contains the information about such subgraphs that is needed to do dynamic programming for  \textsc{Steiner Tree}: for each partition~$P$ of~$Z$ that can be realized by at least one subgraph of~$G_x$, it records the minimum cost of such a subgraph.

    We will use the original algorithm from Bodlaender, Cygan, Kratsch, and Nederlof~\cite{BODLAENDER201586} for all nodes of the decomposition that are not a leaf, nor a leaf-introduce node. For completeness, we will repeat how to compute the table entries for these nodes. The correctness of the formulas follows from the fact that these operations preserve representation, meaning that if $\mathcal{A}'$ represents $\mathcal{A}$, then $f(\mathcal{A}')$ represents $f(\mathcal{A})$. For the proof that the operations used in the recursive formulas preserve representation we refer the reader to the original paper~\cite{BODLAENDER201586}. For the leaf-introduce nodes we will use a new base case that can directly be computed with Lemma~\ref{lem:leafcase}.  
    Note that to ensure that for all internal nodes $x \in \T$ and $Z \subseteq \chi(x)$ we have $|\mathcal{A}'_x(Z)| \le 2^{\Oh(\twtfree(G,K))}$, we will compute a representative set of it right after computing $\mathcal{A}'_x(Z)$ using Proposition~\ref{prop:findrepset} and use it to replace $\mathcal{A}'_x(Z)$. 
    
    \subparagraph*{Introduce vertex node $x$ with child $c$ introducing $v$} When a vertex $v$ is introduced, then either it is going to be part of a partial solution (i.e., $v\in Z$) or not (i.e., $v\not \in Z$). If $v \in Z$, then we will add it as a singleton into each partition of $\mathcal{A}'_c(Z\setminus \{v\})$. This can be interpreted as adding the vertex $v$ (without any incident edges) to any partial solution found in the child bag, using the same boundary vertices, except for $v$. If $v \not \in Z$ and~$v \notin K$, then it should not be used, so we can just copy the value of $\mathcal{A}'_c(Z)$. If~$v \notin Z$ but $v \in K$, then we have~$\mathcal{E}_x(P,Z) = \emptyset$ by definition and may therefore set~$\mathcal{A}_x'(Z) = \emptyset$.
    \[\mathcal{A}'_x(Z) = \begin{cases}
    \{(P \cup \{v\},w): (P,w) \in \mathcal{A}'_c(Z \setminus \{v\})\}  & \text{if } v \in Z \\
    \mathcal{A}'_c(Z) & \text{if } v \not \in Z, v \not \in K\\
    \emptyset & \text{if } v \not \in Z, v \in K \\
    \end{cases}\]

    \subparagraph*{Forget vertex node $x$ with child $c$ forgetting $v$} 
    When $v$ is forgotten in $x$, then we combine the possibilities of whether $v$ was used in a partial solution (corresponding to $\mathcal{A}'_c(Z \cup \{v\})$) or not (corresponding to $\mathcal{A}'_c(Z)$). However, we must ensure that $v$ was not a singleton in a partition if it was used, so that all connected components of partial solutions are connected to the boundary. 
    

    In the following recursive formula we use notation from \cite{BODLAENDER201586} to denote $P_{\downarrow Z}$ as the partition that is the result of taking $P$ and removing all elements not in $Z$ from the sets of the partition, i.e., where $v$ is removed from the partition that it is in.
    \[\mathcal{A}'_x(Z) = \mathcal{A}'_c(Z) \cup \{(P_{\downarrow Z},w): (P,w) \in \mathcal{A}'_c(Z \cup \{v\}) \wedge \{v\} \not \in P\} \]

    \subparagraph*{Introduce edge node $x$ with child $c$ introducing $\{u,v\}$}
     When an edge $\{u,v\}$ is introduced, then either it is going to be part of a partial solution or not. However, it can only be part of a partial solution if both its endpoints are used by the partial solution, i.e., when $u\in Z$ and $v\in Z$. Partial solutions that do not use $\{u,v\}$ correspond to the table entry $\mathcal{A}'_c(Z)$. If $\{u,v\}$ is used, we take all $(P,w)$ in $\mathcal{A}_c'(Z)$ and add the edge. This means that the sets of $P$ containing $u$ and $v$ should be merged since they are now connected, meaning we get partition $P \sqcup \pi_{\{\{u,v\}\}}(Z)$. Moreover, the weight of such a partition increases by $\cost(\{u,v\})$. 
     
    \[\mathcal{A}'_x(Z) = \begin{cases}
    \mathcal{A}'_c(Z) \cup \{(P \sqcup \pi_{\{\{u,v\}\}}(Z), w + \cost(\{u,v\})): (P,w) \in \mathcal{A}'_c(Z)\}  & \text{if $\{u,v\} \subseteq Z$,}\\
    \mathcal{A}'_c(Z) & \text{otherwise.}\\
    \end{cases}\]
    
    \subparagraph*{Join node $x$ with child nodes $c_1$ and $c_2$} The recursively formula in this case corresponds to combining all possible partial solutions, meaning that their connectivities and weights are combined into a new partition-weight combination. 
    \[\mathcal{A}'_x(Z) = \{(P_1 \sqcup P_2, w_1 + w_2): (P_1,w_1) \in \mathcal{A}'_{c_1}(Z) \wedge (P_2,w_2) \in \mathcal{A}'_{c_2}(Z) \}\]

    \subparagraph*{Base case: Leaf introduce node $x$ with child node $c$}
    Let $x$ be a leaf introduce node of $\T$ with child $c \in V(\T)$ that is a leaf of $\T$. We will show how to compute the table entry of $\mathcal{A}'_x(Z)$ for all $Z \subseteq \chi(x)$ in $2^{\Oh(|\chi(x)|)}\cdot \poly(n)$ time. By definition of a nice tree $K$-free-decomposition, node $c$ is the only child of $x$ and $\chi(c) \setminus \chi(x) \subseteq L$. We denote $Y = \chi(c)\setminus \chi(x)$ and remark that $Y \subseteq L$ does not contain any terminals. Moreover, we have $G_x = G[\chi(c)]$. 
    
    For any $Z \subseteq \chi(x)$ such that $K \cap \chi(x) \subseteq Z$ define $\mathcal{F}_x(Z) = \{F \subseteq G[ Z \cup Y]\}$ and note that we can compute a set $\mathcal{R}_x(Z)\subseteq \mathcal{F}_x(Z)$ that represents  $\mathcal{F}_x(Z)$ on $Z$ (according to Definition~\ref{def:represent:subgraph}) of size at most $2^{|Z|-1}$ in $2^{\Oh(|Z|)}\cdot \poly(n)$ time, using Lemma~\ref{lem:leafcase}. We will show that we can set $\mathcal{A}'_x(Z)$ to be:
    \[\mathcal{A}'_x(Z) = \left\{\left(P, \min_{F \in \mathcal{R}_x(Z): \pi_F(Z) = P} \cost(F) \right): P \in \Pi(Z) \wedge \{F \in \mathcal{R}_x(Z): \pi_F(Z) =P\} \not = \emptyset\right \}.\]

    Intuitively, the set~$\mathcal{A}'_x(Z)$ of weighted partitions representing partial solutions of~$G_x$ whose intersection with the current bag~$\chi(x)$ matches~$Z$ is built as follows: for each partition~$P$ of~$Z$ that is realized by the connectivity pattern of at least one subgraph stored in~$\mathcal{R}_x(Z)$, we add that partition together with the minimum cost of any subgraph realizing it to~$\mathcal{A}'_x(Z)$.
    
    To show correctness, we first rewrite the definition of $\mathcal{A}_x(Z)$ for leaf introduce nodes and show that $\mathcal{A}'_x(Z)$ represents $\mathcal{A}_x(Z)$ (according to Definition~\ref{def:represent:partition}). For this purpose, we first rewrite the definition of $\mathcal{E}_x(P,Z)$ for a leaf introduce node $x$. Recall the definition: 
   \[\mathcal{E}_x(P,Z) = \left\{ F \subseteq G_x \colon \begin{aligned}& 
    V(F) \cap \chi(x) = Z, \pi_F(Z)  = P, \\ &\#\mathsf{cc}(F) = |P|, V(G_x)\cap K \subseteq V(F)  \end{aligned}\right\}.\]

    First note that by definition of tree $K$-free-decomposition, there are no terminals in $Y = \chi(c)\setminus\chi(x)$. Hence, for any $F \subseteq G_x$ such that $V(F)\cap \chi(x) = Z$, we have that the requirements $V(G_x) \cap K \subseteq V(F)$ and $K \cap \chi(x) \subseteq Z$ are equivalent, i.e., we can replace the first requirement with the second in the definition. 
    Now let 
    \[\mathcal{E}'_x(P,Z) = \begin{cases} \{ F \subseteq G[Z \cup Y] \colon  
    \pi_F(Z)  = P\} & \text{if } K\cap \chi(x) \subseteq Z, \\ \emptyset &\text{otherwise}. \end{cases}\]
    We will show that $\min\{ \cost(F): F \in \mathcal{E}_x(P,Z)\} = \min\{ \cost(F): F \in \mathcal{E}'_x(P,Z)\}$. First note that $\mathcal{E}_x(P,Z) \subseteq \mathcal{E}'_x(P,Z)$ as in principle we removed the requirements $\#\mathsf{cc}(F) = |P|$ and $V(F) \cap \chi(x) = Z$ from the definition of $\mathcal{E}_x(P,Z)$ to obtain $\mathcal{E}'_x(P,Z)$. 
    
    Now take any $F \in \mathcal{E}'_x(P,Z) \setminus \mathcal{E}_x(P,Z)$. We will prove that there is a $F' \in \mathcal{E}_x(P,Z)$ with equal or smaller cost. 
    Recall that $\pi_F(Z)=P$ does not necessarily imply that $V(F) \cap \chi(x)=Z$, as the definition of $\pi_F(Z)$ adds singleton sets for vertices in $Z\setminus V(F)$. However, we may assume $V(F) \cap \chi(x)  = Z$, as otherwise we can add such a missing vertex from $Z \setminus (V(F)\cap \chi(x))$ to $V(F)$ as an isolated vertex to get to a subgraph that is also in $\mathcal{E}'_x(P,Z)$ with the same cost as $F$. Since~$Y \cap \chi(x) = \emptyset$ by definition and~$F \subseteq G[Z \cup Y]$, the subgraph~$F$ cannot contain any vertex of~$\chi(x) \setminus Z$, so that we indeed obtain~$V(F) \cap \chi(x) = Z$. Hence, we are only left with the case of $\#\mathsf{cc}(F) > |P|$, i.e., there is a connected component of $F$ not intersecting $Z$. Since this connected component cannot contain any terminals (as $Y$ does not contain any terminals), we find that removing the connected component from $F$ results in a subgraph that is of smaller (or equal) cost and part of $\mathcal{E}'_x(P,Z)$ . Hence, we may assume that for the $F \in \mathcal{E}'_x(P,Z)$ of minimum cost, we have $\#\mathsf{cc}(F) = |P|$, i.e., $F \in \mathcal{E}_x(P,Z)$. Hence, we can conclude that $\min\{ \cost(F): F \in \mathcal{E}_x(P,Z)\} = \min\{ \cost(F): F \in \mathcal{E}'_x(P,Z)\}$. As a consequence, we can replace $\mathcal{E}_x(P,Z)$ with $\mathcal{E}'_x(P,Z)$ in the definition of $\mathcal{A}_x(Z)$ without changing its contents.

    Recall that~$\mathcal{F}_x(Z) = \{F \subseteq G[ Z \cup Y]\}$, so that $\mathcal{E}'_x(P,Z) = \{F \in \mathcal{F}_x(Z): \pi_F(Z) = P\}$ whenever $K \cap \chi(x)\subseteq Z$. Hence, for any $Z$ such that $K \cap \chi(x) \subseteq Z$ we actually have:
    \[\mathcal{A}_x(Z) = \left\{\left(P, \min_{F \in \mathcal{F}_x(Z): \pi_F(Z) = P} \cost(F) \right): P \in \Pi(Z) \wedge \{F \in \mathcal{F}_x(Z): \pi_F(Z) =P\} \not = \emptyset\right \}.\]

    Now that we've rewritten $\mathcal{A}_x(Z)$, it is time to prove that $\mathcal{A}'_x(Z)$ represents $\mathcal{A}_x(Z)$, i.e., we show that 
    for all $Q \in \Pi(Z)$:
    \[\min\{w \colon (P,w) \in \mathcal{A}_x(Z) \wedge P\sqcup Q = \{Z\}\} = \min\{w \colon (P,w) \in \mathcal{A}'_x(Z) \wedge P\sqcup Q = \{Z\}\}.\] 
    Note that $\mathcal{A}'_x(Z) \subseteq \mathcal{A}_x(Z)$ by the new definition of $\mathcal{A}_x(Z)$, as $\mathcal{R}_x(Z) \subseteq \mathcal{F}_x(Z)$. Hence we have that the left-hand side is smaller or equal to the right-hand side of the equality. 

    For the other direction, take any $(P,w)\in\mathcal{A}_x(Z)$ that attains the minimum value. Hence, there exists a $F \in \mathcal{F}_x(Z)$ such that $\pi_F(Z) \sqcup Q = \{Z\}$ and $\cost(F)=w$. Because $\mathcal{R}_x(Z)$ represents $\mathcal{F}_x(Z)$ on $Z$ (according to Definition~\ref{def:represent:subgraph}), we know that there exists a $F' \in \mathcal{R}_x(Z)$ such that $\pi_{F'}(Z) \sqcup Q = \{Z\}$ and $\cost(F') = w$. Therefore, there exists $(\pi_{F'}(Z),w) \in \mathcal{A}'_x(Z)$ with $\pi_{F'}(Z)\sqcup Q = \{Z\}$, i.e., we find that 
       \[\min\{w \colon (P,w) \in \mathcal{A}_x(Z) \wedge P\sqcup Q = \{Z\}\} = \min\{w \colon (P,w) \in \mathcal{A}'_x(Z) \wedge P\sqcup Q = \{Z\}\}.\]

    To summarize, the algorithm will compute the values of $\mathcal{A}_x'(Z)$ using the tree $K$-free decomposition and the recursive formulas given above. Each time after a value of $\mathcal{A}_x'(Z)$ is computed, the algorithm will apply Proposition~\ref{prop:findrepset} to find a representative set and use it to replace $\mathcal{A}'(Z)$. This last step ensures that $\mathcal{A}'(Z) \le 2^{\twtfree(G,K)}$ throughout the algorithm. The algorithm will then return the value $\min_{Z\subseteq \chi(r)}\{\mathcal{A}'_r(Z)\}$ where $r$ denotes the root of the tree decomposition. 
    
    For the running time analysis, note that computing the values for leaf-introduce nodes takes $2^{\Oh(\twtfree(G,K))}\cdot \poly(n)$ time, using Lemma~\ref{lem:leafcase} and the definition of $\mathcal{A}'_x(Z)$ for such nodes directly. Moreover, computing the values for internal vertices takes $2^{\Oh(\twtfree(G,K))}\cdot \poly(n)$ time using the recursive formulas for other internal nodes, as we ensure that $|\mathcal{A}'_x(Z)| \le 2^{\twtfree(G,K)}$ for all $x$ and $Z\subseteq \chi(x)$ throughout the algorithm. Therefore, the total algorithm runs in time bounded by $2^{\Oh(\twtfree(G,K))}\cdot\poly(n)$.  
\end{proof}

\section{Conclusion} \label{sec:conclusion}
In this paper, we show that it is possible to design FPT algorithms for \textsc{Steiner Tree} when the considered parameter is structurally smaller than the number of terminals. For this purpose, we have extended the definitions of $\mathcal{H}$-free treewidth and $\mathcal{H}$-elimination distance to terminal-free variants of these parameters. 

Our polynomial-space algorithm for \textsc{Steiner Tree} parameterized by multiway cut is slightly superexponential, due to the number of different $S$-connecting systems. Is it possible to design an algorithm for \textsc{Steiner Tree} parameterized by $|S|$ that uses both polynomial space \emph{and} single-exponential FPT time?

Another direction for future research is to apply such terminal-aware parameterizations to other problems that rely on terminals. An obvious example is \textsc{Multiway Cut} itself. Another problem to consider is \textsc{Shortest $K$-Cycle}, where one is given an undirected graph and a set of terminals $K$ and asked to find a simple cycle that goes through all terminals. Bj\"orklund, Husfeldt, and Taslaman~\cite{bjorklund2012shortest} give a $2^{|K|}\poly(n)$ time algorithm for this problem. Is it fixed-parameter tractable parameterized by multiway cut, or even $K$-free treewidth? The same question can be asked for the \textsc{Subset Traveling Salesperson Problem} (\textsc{Subset TSP}), which asks for a minimum-weight tour visiting a given subset~$K$ of terminal vertices. Many techniques developed for \textsc{Steiner Tree} parameterized by treewidth also apply to \textsc{Hamiltonian Cycle} and other variants of TSP~\cite{BODLAENDER201586,cygan2022solving}; we expect that our techniques for terminal-aware parameterizations can be extended for \textsc{Subset TSP}.

\bibliographystyle{plainurl}
\bibliography{references}

\begin{thebibliography}{10}

\bibitem{Argawal2022}
Akanksha Agrawal, Lawqueen Kanesh, Daniel Lokshtanov, Fahad Panolan, M.~S.
  Ramanujan, Saket Saurabh, and Meirav Zehavi.
\newblock Deleting, eliminating and decomposing to hereditary classes are all
  {FPT}-equivalent.
\newblock In {\em Proceedings of SODA 2022}, pages 1976--2004, 2022.
\newblock \href {https://doi.org/10.1137/1.9781611977073.79}
  {\path{doi:10.1137/1.9781611977073.79}}.

\bibitem{DistanceFromTriviality}
Akanksha Agrawal and M.~S. Ramanujan.
\newblock Distance from triviality 2.0: {H}ybrid parameterizations.
\newblock In {\em Proc. 33rd IWOCA}, volume 13270 of {\em Lecture Notes in
  Computer Science}, pages 3--20. Springer, 2022.
\newblock \href {https://doi.org/10.1007/978-3-031-06678-8_1}
  {\path{doi:10.1007/978-3-031-06678-8_1}}.

\bibitem{bjorklund2012shortest}
Andreas Bj{\"o}rklund, Thore Husfeldt, and Nina Taslaman.
\newblock Shortest cycle through specified elements.
\newblock In {\em Proc. 23rd SODA}, pages 1747--1753. SIAM, 2012.
\newblock \href {https://doi.org/10.1137/1.9781611973099.139}
  {\path{doi:10.1137/1.9781611973099.139}}.

\bibitem{BODLAENDER201586}
Hans~L. Bodlaender, Marek Cygan, Stefan Kratsch, and Jesper Nederlof.
\newblock Deterministic single exponential time algorithms for connectivity
  problems parameterized by treewidth.
\newblock {\em Information and Computation}, 243:86--111, 2015.
\newblock \href {https://doi.org/10.1016/j.ic.2014.12.008}
  {\path{doi:10.1016/j.ic.2014.12.008}}.

\bibitem{eliminationbulian2016graph}
Jannis Bulian and Anuj Dawar.
\newblock Graph isomorphism parameterized by elimination distance to bounded
  degree.
\newblock {\em Algorithmica}, 75(2):363--382, 2016.
\newblock \href {https://doi.org/10.1007/S00453-015-0045-3}
  {\path{doi:10.1007/S00453-015-0045-3}}.

\bibitem{vertex-deletion}
Leizhen Cai.
\newblock Fixed-parameter tractability of graph modification problems for
  hereditary properties.
\newblock {\em Information Processing Letters}, 58(4):171--176, 1996.
\newblock \href {https://doi.org/10.1016/0020-0190(96)00050-6}
  {\path{doi:10.1016/0020-0190(96)00050-6}}.

\bibitem{MWCinFPT}
Jianer Chen, Yang Liu, and Songjian Lu.
\newblock An improved parameterized algorithm for the minimum node multiway cut
  problem.
\newblock {\em Algorithmica}, 55(1):1--13, 2009.
\newblock \href {https://doi.org/10.1007/S00453-007-9130-6}
  {\path{doi:10.1007/S00453-007-9130-6}}.

\bibitem{cygan2015parameterized}
Marek Cygan, Fedor~V. Fomin, Lukasz Kowalik, Daniel Lokshtanov, D{\'{a}}niel
  Marx, Marcin Pilipczuk, Michal Pilipczuk, and Saket Saurabh.
\newblock {\em Parameterized Algorithms}.
\newblock Springer, 2015.
\newblock \href {https://doi.org/10.1007/978-3-319-21275-3}
  {\path{doi:10.1007/978-3-319-21275-3}}.

\bibitem{cygan2022solving}
Marek Cygan, Jesper Nederlof, Marcin Pilipczuk, Micha{\l} Pilipczuk, Johan~M.M.
  Van~Rooij, and Jakub~Onufry Wojtaszczyk.
\newblock Solving connectivity problems parameterized by treewidth in single
  exponential time.
\newblock {\em ACM Transactions on Algorithms (TALG)}, 18(2):1--31, 2022.
\newblock \href {https://doi.org/10.1145/3506707} {\path{doi:10.1145/3506707}}.

\bibitem{DIMACS}
{11th DIMACS Implementation Challenge}, 2014.
\newblock Accessed 21 Nov 2023.
\newblock URL: \url{https://dimacs11.zib.de/}.

\bibitem{dreyfus1971steiner}
Stuart~E. Dreyfus and Robert~A. Wagner.
\newblock The {S}teiner problem in graphs.
\newblock {\em Networks}, 1(3):195--207, 1971.
\newblock \href {https://doi.org/10.1002/net.3230010302}
  {\path{doi:10.1002/net.3230010302}}.

\bibitem{HtreewidthEIBEN202157}
Eduard Eiben, Robert Ganian, Thekla Hamm, and O~joung Kwon.
\newblock Measuring what matters: A hybrid approach to dynamic programming with
  treewidth.
\newblock {\em Journal of Computer and System Sciences}, 121:57--75, 2021.
\newblock \href {https://doi.org/10.1016/j.jcss.2021.04.005}
  {\path{doi:10.1016/j.jcss.2021.04.005}}.

\bibitem{fomin2019parameterized}
Fedor~V Fomin, Petteri Kaski, Daniel Lokshtanov, Fahad Panolan, and Saket
  Saurabh.
\newblock Parameterized single-exponential time polynomial space algorithm for
  {S}teiner tree.
\newblock {\em SIAM Journal on Discrete Mathematics}, 33(1):327--345, 2019.
\newblock \href {https://doi.org/10.1137/17M1140030}
  {\path{doi:10.1137/17M1140030}}.

\bibitem{fuchs2007dynamic}
Bernhard Fuchs, Walter Kern, D~Molle, Stefan Richter, Peter Rossmanith, and
  Xinhui Wang.
\newblock Dynamic programming for minimum {S}teiner trees.
\newblock {\em Theory of Computing Systems}, 41:493--500, 2007.
\newblock \href {https://doi.org/10.1007/S00224-007-1324-4}
  {\path{doi:10.1007/S00224-007-1324-4}}.

\bibitem{distancetotriviality1guo2004structural}
Jiong Guo, Falk H{\"u}ffner, and Rolf Niedermeier.
\newblock A structural view on parameterizing problems: Distance from
  triviality.
\newblock In {\em Proc. 1st IWPEC}, pages 162--173. Springer, 2004.
\newblock \href {https://doi.org/10.1007/978-3-540-28639-4_15}
  {\path{doi:10.1007/978-3-540-28639-4_15}}.

\bibitem{biologyideker2002discovering}
Trey Ideker, Owen Ozier, Benno Schwikowski, and Andrew~F Siegel.
\newblock Discovering regulatory and signalling circuits in molecular
  interaction networks.
\newblock {\em Bioinformatics}, 18:S233--S240, 2002.
\newblock \href {https://doi.org/10.1093/bioinformatics/18.suppl_1.S233}
  {\path{doi:10.1093/bioinformatics/18.suppl_1.S233}}.

\bibitem{ImpagliazzoP01}
Russell Impagliazzo and Ramamohan Paturi.
\newblock On the complexity of k-sat.
\newblock {\em J. Comput. Syst. Sci.}, 62(2):367--375, 2001.
\newblock \href {https://doi.org/10.1006/JCSS.2000.1727}
  {\path{doi:10.1006/JCSS.2000.1727}}.

\bibitem{jansen2021vertex}
Bart M.~P. Jansen, Jari J.~H. de~Kroon, and Micha{\l} W{\l}odarczyk.
\newblock Vertex deletion parameterized by elimination distance and even less.
\newblock In {\em Proc. 53rd {STOC}}, pages 1757--1769. {ACM}, 2021.
\newblock \href {https://doi.org/10.1145/3406325.3451068}
  {\path{doi:10.1145/3406325.3451068}}.

\bibitem{JansenKW23}
Bart M.~P. Jansen, Jari J.~H. de~Kroon, and Michal Wlodarczyk.
\newblock 5-approximation for {$\mathcal{H}$}-treewidth essentially as fast as
  {$\mathcal{H}$}-deletion parameterized by solution size.
\newblock In {\em Proc. 31st ESA}, volume 274 of {\em LIPIcs}, pages
  66:1--66:16. Schloss Dagstuhl - Leibniz-Zentrum f{\"{u}}r Informatik, 2023.
\newblock \href {https://doi.org/10.4230/LIPIcs.ESA.2023.66}
  {\path{doi:10.4230/LIPIcs.ESA.2023.66}}.

\bibitem{Karp1972}
Richard~M. Karp.
\newblock Reducibility among combinatorial problems.
\newblock In {\em Proceedings of a symposium on the Complexity of Computer
  Computations}, The {IBM} Research Symposia Series, pages 85--103. Plenum
  Press, New York, 1972.
\newblock \href {https://doi.org/10.1007/978-1-4684-2001-2_9}
  {\path{doi:10.1007/978-1-4684-2001-2_9}}.

\bibitem{kloks1994treewidth}
Ton Kloks.
\newblock {\em Treewidth: computations and approximations}, volume 842.
\newblock Springer Science \& Business Media, 1994.
\newblock \href {https://doi.org/10.1007/BFb0045375}
  {\path{doi:10.1007/BFb0045375}}.

\bibitem{circuitlengauer2012combinatorial}
Thomas Lengauer.
\newblock {\em Combinatorial algorithms for integrated circuit layout}.
\newblock XApplicable Theory in Computer Science. Springer Science \& Business
  Media, 2012.
\newblock \href {https://doi.org/10.1007/978-3-322-92106-2}
  {\path{doi:10.1007/978-3-322-92106-2}}.

\bibitem{survey}
Ivana Ljubić.
\newblock Solving {S}teiner trees: Recent advances, challenges, and
  perspectives.
\newblock {\em Networks}, 77(2):177--204, 2021.
\newblock \href {https://doi.org/10.1002/net.22005}
  {\path{doi:10.1002/net.22005}}.

\bibitem{lokshtanov2010saving}
Daniel Lokshtanov and Jesper Nederlof.
\newblock Saving space by algebraization.
\newblock In {\em Proc. 42nd STOC}, pages 321--330. {ACM}, 2010.
\newblock \href {https://doi.org/10.1145/1806689.1806735}
  {\path{doi:10.1145/1806689.1806735}}.

\bibitem{nederlof2013polyspace}
Jesper Nederlof.
\newblock Fast polynomial-space algorithms using inclusion-exclusion.
\newblock {\em Algorithmica}, 65(4):868--884, 2013.
\newblock \href {https://doi.org/10.1007/S00453-012-9630-X}
  {\path{doi:10.1007/S00453-012-9630-X}}.

\bibitem{treedepth_nesetril2012bounded}
Jaroslav Ne{\v{s}}et{\v{r}}il and Patrice Ossona~de Mendez.
\newblock Chapter 6: Bounded height trees and tree-depth.
\newblock In {\em Sparsity: Graphs, Structures, and Algorithms}, volume~28 of
  {\em Algorithms and Combinatorics}, pages 115--144. Springer, 2012.
\newblock \href {https://doi.org/10.1007/978-3-642-27875-4}
  {\path{doi:10.1007/978-3-642-27875-4}}.

\bibitem{PACE}
{PACE 2018}, 2018.
\newblock Accessed 21 Nov 2023.
\newblock URL: \url{https://pacechallenge.org/2018/}.

\bibitem{telecommunication}
Mauricio~G.C. Resende and Panos~M. Pardalos.
\newblock {\em Handbook of optimization in telecommunications}.
\newblock Springer Science \& Business Media, 2008.
\newblock \href {https://doi.org/10.1007/978-0-387-30165-5}
  {\path{doi:10.1007/978-0-387-30165-5}}.

\bibitem{objectrussakovsky2010steiner}
Olga Russakovsky and Andrew~Y. Ng.
\newblock A {S}teiner tree approach to efficient object detection.
\newblock In {\em Proc. 23rd CVPR}, pages 1070--1077. {IEEE} Computer Society,
  2010.
\newblock \href {https://doi.org/10.1109/CVPR.2010.5540097}
  {\path{doi:10.1109/CVPR.2010.5540097}}.

\bibitem{Roozendaal23}
Tom van Roozendaal.
\newblock Fixed-parameter tractable algorithms for refined parameterizations of
  graph problems.
\newblock Master's thesis, Eindhoven University of Technology, 2023.
\newblock URL:
  \url{https://research.tue.nl/en/studentTheses/fixed-parameter-tractable-algorithms-for-refined-parameterization}.

\end{thebibliography}

\end{document}